\documentclass[11pt]{article}

\usepackage{times}
\usepackage[paperwidth=199.8mm,
paperheight=297mm,centering,hmargin=25mm,vmargin=25mm]{geometry}
\usepackage{authblk} 
\usepackage[bottom]{footmisc} 

\usepackage[titletoc,title]{appendix} 
\usepackage{changepage}

\usepackage{hyperref}
\hypersetup{colorlinks=true,citecolor=myblue,linkcolor=myblue,
filecolor=myblue,urlcolor=myblue,breaklinks=mytrue}
\usepackage{url}

\usepackage[dvipsnames]{xcolor}
\usepackage{color}
\usepackage{framed}

\definecolor{shadecolor}{rgb}{0.9,0.9,0.9}
\definecolor{mylightgray}{RGB}{220,220,220}

\definecolor{myblue}{RGB}{0, 68, 116}
\definecolor{mycyan}{RGB}{0, 97, 91}
\definecolor{mygreen}{RGB}{2, 102, 1}
\definecolor{myorange}{RGB}{240, 102, 0}
\definecolor{myred}{RGB}{172, 23, 0}
\definecolor{mymagenta}{RGB}{140,16,73}

\usepackage{mathtools}
\usepackage{amsmath}
\usepackage{amssymb}
\usepackage[shortlabels]{enumitem}
\usepackage{graphicx,epic,eepic,epsfig,latexsym,verbatim}
\usepackage{dsfont}
\usepackage{mathrsfs}
\usepackage{bm}
\usepackage{amsbsy}

\usepackage{subcaption}
\usepackage{caption}
\usepackage{float}
\usepackage{tikz}
\usepackage{relsize}
\usepackage{pgfplots}

\usepackage{amsthm}
\usepackage{tcolorbox}
\tcbuselibrary{skins}
\tcbuselibrary{breakable}

\makeatletter
\newcommand{\DefineColoredTheoremStyle}[2]{%
  \newtheoremstyle{#1}
    {3pt}{3pt}
    {\itshape}
    {}
    {\bfseries\color{#2}}
    {\;}
    { }
    {\thmname{##1}\thmnumber{ ##2}\textnormal{\thmnote{ (##3)}}}
}
\makeatother

\DefineColoredTheoremStyle{lemmstyle}{black}
\DefineColoredTheoremStyle{propstyle}{black}
\DefineColoredTheoremStyle{thmstyle}{black}
\DefineColoredTheoremStyle{corostyle}{black}
\DefineColoredTheoremStyle{conjstyle}{black}
\DefineColoredTheoremStyle{problstyle}{black}
\DefineColoredTheoremStyle{plainstyle}{black}

\theoremstyle{lemmstyle}\newtheorem{lemma}{Lemma}[section]
\theoremstyle{propstyle}\newtheorem{prop}{Proposition}[section]
\theoremstyle{thmstyle}\newtheorem{theorem}{Theorem}[section]
\theoremstyle{corostyle}\newtheorem{corollary}{Corollary}[section]
\theoremstyle{conjstyle}\newtheorem{conjecture}{Conjecture}[section]
\theoremstyle{problstyle}\newtheorem{problem}{Problem}[section]
\theoremstyle{plainstyle}\newtheorem{definition}{Definition}[section]
\theoremstyle{plainstyle}\newtheorem{assumption}{Assumption}[section]
\theoremstyle{plainstyle}\newtheorem{example}{Example}[section]
\theoremstyle{plainstyle}\newtheorem{remark}{Remark}[section]

\tcbset{
    commonstyle/.style n args={1}{
        oversize=0.5mm,
        colback={#1!8},
        colframe={#1!10},
        boxrule=1pt,
        boxsep=2mm,
        left=1mm,
        right=1mm,
        top=2mm,
        bottom=2mm,
        breakable
    }
}

\newenvironment{boxdefinition}{%
    \begin{tcolorbox}[commonstyle={black}]%
        \begin{definition}%
}{%
        \end{definition}%
    \end{tcolorbox}%
}

\newenvironment{boxassumption}{%
    \begin{tcolorbox}[commonstyle={mylightgray}]%
        \begin{assumption}%
}{%
    \end{assumption}%
  \end{tcolorbox}%
}

\newenvironment{boxlemma}{%
  \begin{tcolorbox}[commonstyle={myblue}]%
    \begin{lemma}%
}{%
    \end{lemma}%
  \end{tcolorbox}%
}

\newenvironment{boxtheorem}{%
  \begin{tcolorbox}[commonstyle={myorange}]%
    \begin{theorem}%
}{%
    \end{theorem}%
  \end{tcolorbox}%
}

\newenvironment{boxexample}{%
  \begin{tcolorbox}[commonstyle={mycyan}]%
    \begin{example}%
}{%
    \end{example}%
  \end{tcolorbox}%
}

\usepackage{colortbl}
\usepackage{pifont} 
\usepackage{booktabs}
\usepackage{makecell} 
\usepackage{diagbox}
\usepackage{multirow}

\newcommand{\nc}{\newcommand}
\nc{\rnc}{\renewcommand}

\nc{\<}{\langle}
\rnc{\>}{\rangle}
\nc{\bra}[1]{\langle#1|}
\nc{\ket}[1]{|#1\rangle}
\nc{\ketbra}[2]{|#1\rangle\!\langle#2|}
\nc{\braket}[2]{\langle#1|#2\rangle}
\nc{\braandket}[3]{\langle #1|#2|#3\rangle}
\nc{\proj}[1]{| #1\rangle\!\langle #1 |}
\nc{\avg}[1]{\langle#1\rangle}

\nc{\rank}{\operatorname{Rank}}
\nc{\id}{{\operatorname{id}}}
\nc{\supp}{{\operatorname{supp}}}
\nc{\smfrac}[2]{\mbox{$\frac{#1}{#2}$}}
\nc{\tr}{\operatorname{Tr}}
\nc{\ox}{\otimes}
\nc{\floor}[1]{\lfloor #1 \rfloor}
\nc{\trans}{\mathsf T}
\nc{\img}{\mathbf{i}}

\nc{\cA}{{\cal A}}
\nc{\cB}{{\cal B}}
\nc{\cC}{{\cal C}}
\nc{\cD}{{\cal D}}
\nc{\cE}{{\cal E}}
\nc{\cF}{{\cal F}}
\nc{\cG}{{\cal G}}
\nc{\cH}{{\cal H}}
\nc{\cI}{{\cal I}}
\nc{\cJ}{{\cal J}}
\nc{\cK}{{\cal K}}
\nc{\cL}{{\cal L}}
\nc{\cM}{{\cal M}}
\nc{\cN}{{\cal N}}
\nc{\cO}{{\cal O}}
\nc{\cP}{{\cal P}}
\nc{\cQ}{{\cal Q}}
\nc{\cR}{{\cal R}}
\nc{\cS}{{\cal S}}
\nc{\cT}{{\cal T}}
\nc{\cV}{{\cal V}}
\nc{\cU}{{\cal U}}
\nc{\cX}{{\cal X}}
\nc{\cY}{{\cal Y}}
\nc{\cZ}{{\cal Z}}
\nc{\cW}{{\cal W}}

\nc{\RR}{{{\mathbb R}}}
\nc{\CC}{{{\mathbb C}}}
\nc{\FF}{{{\mathbb F}}}
\nc{\NN}{{{\mathbb N}}}
\nc{\ZZ}{{{\mathbb Z}}}
\nc{\QQ}{{{\mathbb Q}}}
\nc{\UU}{{{\mathbb U}}}
\nc{\EE}{{{\mathbb E}}}

\nc{\bH}{{\mathfrak{H}}}

\nc{\sK}{{{\mathscr{K}}}}
\nc{\sS}{{{\mathscr{S}}}}
\nc{\sT}{{{\mathscr{T}}}}
\nc{\sA}{{{\mathscr{A}}}}
\nc{\sB}{{{\mathscr{B}}}}
\nc{\sC}{{{\mathscr{C}}}}
\nc{\sE}{{{\mathscr{E}}}}
\nc{\sL}{{{\mathscr{L}}}}
\nc{\sG}{{{\mathscr{G}}}}
\nc{\sF}{{{\mathscr{F}}}}
\nc{\sI}{{{\mathscr{I}}}}
\nc{\sN}{{{\mathscr{N}}}}
\nc{\sM}{{{\mathscr{M}}}}
\nc{\bsC}{\pmb{\mathscr{C}}}

\nc{\Choi}{Choi-Jamio\l{}kowski }
\nc{\reg}{\infty}
\nc{\amo}{\text{\rm amo}}
\nc{\Renyi}{R\'{e}nyi }

\nc{\conv}{\operatorname{conv}}
\nc{\cvxset}{\mathscr{C}}

\nc{\RM}{{{\mathscr{R}}}}

\nc{\END}{\operatorname{End}}
\nc{\PERM}{\mathfrak{\sigma}}

\nc{\Cone}{\text{\rm Cone}}
\nc{\sep}{{\SEP}}

\nc{\DD}{{{\mathbb D}}}
\nc{\BS}{{\scriptscriptstyle \rm {BS}}}
\nc{\Sand}{{\scriptscriptstyle  \rm S}}
\nc{\Petz}{{\scriptscriptstyle  \rm P}}
\nc{\Hypo}{{\scriptscriptstyle  \rm H}}
\nc{\Meas}{{\scriptscriptstyle \rm M}}
\nc{\Proj}{{{\scriptscriptstyle \rm P}}}

\nc{\suchthat}{\text{\rm s.t.}}

\nc{\pl}{{\scalebox{0.7}{+}}}
\nc{\HERM}{\mathscr{H}}
\nc{\PSD}{\HERM_{\pl}}
\nc{\PD}{\HERM_{\pl\pl}}
\nc{\density}{\mathscr{D}}
\nc{\subdensity}{\mathscr{D}_\bullet}

\nc{\polarPSD}[1]{{#1}_{\pl}^{\circ}}
\nc{\polarPSDre}[1]{{#1}_{\pl}^{\star}}
\nc{\polarPD}[1]{{#1}_{\pl\pl}^{\circ}}

\nc{\PPT}{\text{\rm PPT}}
\nc{\Rains}{\text{\rm Rains}}
\nc{\WD}{\text{\rm WD}}
\nc{\SEP}{\text{\rm SEP}}
\nc{\PSEP}{\text{\rm PSEP}}
\nc{\CPTP}{\text{\rm CPTP}}
\nc{\POVM}{\text{\rm POVM}}
\nc{\PVM}{\text{\rm PVM}}
\nc{\CP}{\text{\rm CP}}
\nc{\adv}{\text{\rm adv}}
\nc{\spec}{\text{\rm spec}}
\nc{\poly}{\text{\rm poly}}
\nc{\End}{\operatorname{End}}
\nc{\Par}{\operatorname{Par}}
\nc{\RNG}{\operatorname{RNG}}
\nc{\STAB}{\text{\rm STAB}}
\nc{\epi}{\boldsymbol{\operatorname{epi}}}
\nc{\op}{\boldsymbol{\operatorname{op}}}

\makeatletter
\newcommand*\rel@kern[1]{\kern#1\dimexpr\macc@kerna}
\newcommand*\widebar[1]{%
  \begingroup
  \def\mathaccent##1##2{%
    \rel@kern{0.8}%
    \overline{\rel@kern{-0.8}\macc@nucleus\rel@kern{0.2}}%
    \rel@kern{-0.2}%
  }%
  \macc@depth\@ne
  \let\math@bgroup\@empty \let\math@egroup\macc@set@skewchar
  \mathsurround\z@ \frozen@everymath{\mathgroup\macc@group\relax}%
  \macc@set@skewchar\relax
  \let\mathaccentV\macc@nested@a
  \macc@nested@a\relax111{#1}%
  \endgroup
}
\makeatother

\begin{document}

\title{\Large \textbf{Generalized quantum Chernoff bound}}

\author[1]{Kun Fang \thanks{kunfang@cuhk.edu.cn}}
\author[1,2,3]{Masahito Hayashi \thanks{hmasahito@cuhk.edu.cn}}
\affil[1]{\small School of Data Science, The Chinese University of Hong Kong, Shenzhen,\protect\\  Guangdong, 518172, China}
\affil[2]{International Quantum Academy, Futian District, Shenzhen 518048, China}
\affil[3]{Graduate School of Mathematics, Nagoya University, Chikusa-ku, Nagoya 464–8602, Japan}


\date{\today}

\maketitle

\vspace{-0.2cm}
\begin{abstract}
We consider the task of distinguishing whether a quantum system is prepared in a state from one of several sets of quantum states. Assuming their convexity and stability under tensor product, we prove that the optimal error exponent for discrimination is precisely given by the regularized quantum Chernoff divergence between the sets, thereby establishing a generalized quantum Chernoff bound for the discrimination of multiple sets of quantum states. This extends the classical and quantum Chernoff bounds to the general setting of composite and correlated quantum hypotheses. Furthermore, leveraging minimax theorems, we show that discriminating between sets of quantum states is no harder than discriminating between their worst-case elements in terms of error probability. This implies the existence of an optimal state-agnostic test that achieves the minimum error probability for all states in the sets, matching the performance of the optimal state-dependent test for the most difficult pair of states. We provide explicit characterizations of the optimal state-agnostic test in the binary composite case. 
Finally, we show that the maximum overlap between a pure state and a set of free states, a quantity that frequently arises in quantum resource theories, is equal to the quantum Chernoff divergence between the sets, thereby providing an operational interpretation of this quantity in the context of symmetric hypothesis testing.
\end{abstract}

\setcounter{tocdepth}{2}
\tableofcontents

\section{Introduction}

Hypothesis testing is a fundamental method for deciding between competing explanations of observed data. It provides a rigorous framework for making decisions under uncertainty and is central to statistics, information theory, and many applied fields such as signal processing, machine learning, and experimental physics~\cite{lehmann-2022}. The main goal of hypothesis testing is to determine which of several possible models or probability distributions best describes a set of observations, while keeping the probability of making an error as low as possible.
In the classical setting, Chernoff showed that, when many independent samples are available, the probability of error decreases exponentially with the number of samples. The rate at which this error probability decays is precisely given by the celebrated Chernoff bound~\cite{chernoff1952measure}, which quantifies the fundamental limit of distinguishability between two hypotheses. 

\paragraph{Quantum Chernoff bound between two quantum states.} The Chernoff bound has been extended to the quantum setting, where the task is to distinguish between two quantum hypotheses: the system is prepared either in state $\rho_1$ (the null hypothesis) or in state $\rho_2$ (the alternative hypothesis). In the Bayesian framework, each hypothesis is assigned a prior probability, denoted by $\pi_1$ and $\pi_2$. 
Operationally, discrimination is performed using a two-outcome positive operator-valued measure (POVM) $\{M_1, M_2\}$, where $M_1+M_2 = I$. According to Born's rule in quantum mechanics, the average error probability for this measurement is given by
\begin{align}
P_e(\{M_1,M_2\}, \{\rho_1, \rho_2\}) := \pi_1 \tr[\rho_1(I-M_1)] + \pi_2 \tr[\rho_2(I-M_2)].
\end{align}
The objective is to minimize the average error probability over all possible POVMs:
\begin{align}
P_{e, \min}(\{\rho_1, \rho_2\}) := \inf \Big\{P_e(\{M_1,M_2\}, \{\rho_1, \rho_2\}): \{M_1,M_2\} \in \POVM \Big\}.
\end{align}

A landmark result by Audenaert et al.~\cite{audenaert2007discriminating} and Nussbaum and Szkoła~\cite{Nussbaum2009} established that, in the asymptotic regime, the optimal error probability decays exponentially with the number of copies, with the optimal exponent given by
\begin{align}\label{eq: chernoff theorem}
    \lim_{n \to \infty} -\frac{1}{n} \log P_{e,\min}(\{\rho_1^{\otimes n}, \rho_2^{\otimes n}\}) = C(\rho_1 \| \rho_2),
\end{align}
where the quantum Chernoff divergence is defined as
\begin{align}\label{eq: definition of Chernoff divergence}
    C(\rho_1 \| \rho_2) := \max_{0 \leq \alpha \leq 1}  -\log Q_\alpha(\rho_1 \| \rho_2), \quad \text{with} \quad Q_\alpha(\rho_1 \| \rho_2) := \tr[\rho_1^\alpha \rho_2^{1-\alpha}].
\end{align}
This result, known as the quantum Chernoff bound, characterizes the optimal exponential rate at which the error probability decays when discriminating between two independent and identically distributed (i.i.d.) quantum states.

\paragraph{Quantum Chernoff bound for multiple quantum states.} The quantum Chernoff bound has also been extended to the discrimination of multiple quantum hypotheses. Consider $r$ quantum states $\{\rho_i\}_{i=1}^r$ with prior probabilities $\{\pi_i\}_{i=1}^r$. Let $\{M_i\}_{i=1}^r$ be a POVM, i.e., a collection of positive semidefinite operators that sum to the identity. The average error probability for discriminating among these $r$ quantum states is given by
\begin{align}
    P_{e}(\{M_i\}_{i=1}^r, \{\rho_i\}_{i=1}^r) := \sum_{i=1}^r \pi_i \tr[\rho_i (I-M_i)].
\end{align}
Optimizing over all POVMs, the minimum error probability is defined as
\begin{align}\label{eq: Pemin multiple states}
    P_{e,\min}(\{\rho_i\}_{i=1}^r) := \inf \Big\{ P_{e}(\{M_i\}_{i=1}^r, \{\rho_i\}_{i=1}^r): \{M_i\}_{i=1}^r \in \POVM \Big\}.
\end{align}
Li showed that~\cite{li2016discriminating}, in the asymptotic regime,
\begin{align}\label{eq: chernoff theorem for multiple states}
    \lim_{n \to \infty} -\frac{1}{n} \log P_{e,\min}(\{\rho_i^{\otimes n}\}_{i=1}^r) = \min_{i \neq j} C(\rho_i\|\rho_j),
\end{align}
where $C(\rho_i\|\rho_j)$ is the quantum Chernoff divergence between states $\rho_i$ and $\rho_j$. This result generalizes the binary quantum Chernoff bound to the case of multiple quantum states, showing that the optimal error exponent is determined by the pair of states that are hardest to distinguish.

\paragraph{Quantum Chernoff bound for multiple sets of quantum states.} In many practical scenarios, the states to be distinguished are not fully specified (i.e., \emph{composite hypotheses}~\cite{brandao2010generalization,brandao2020adversarial,berta2021composite,Mosonyi_2022,hayashi2024generalized,hayashi2016correlation,DWH25}), such as in adversarial or black-box settings~\cite{fang2025adversarial,watanabe2024black}, and may exhibit correlations that preclude a simple tensor product structure (i.e., \emph{correlated hypotheses}~\cite{hiai2007large,hiai2008error,mosonyi2015two,hayashi2025entanglement}). In this context, the task is to discriminate between multiple \emph{sets of correlated quantum states} (see Figure~\ref{fig: hypothesis testing demo}). That is, a tester receives samples prepared according to one of the sets $\sC_{r,n}$, and determines, via a quantum measurement $\{M_{i,n}\}_{i=1}^r$, from which set the samples originate.

\begin{figure}[H]
    \centering
    \includegraphics[width=0.8\textwidth]{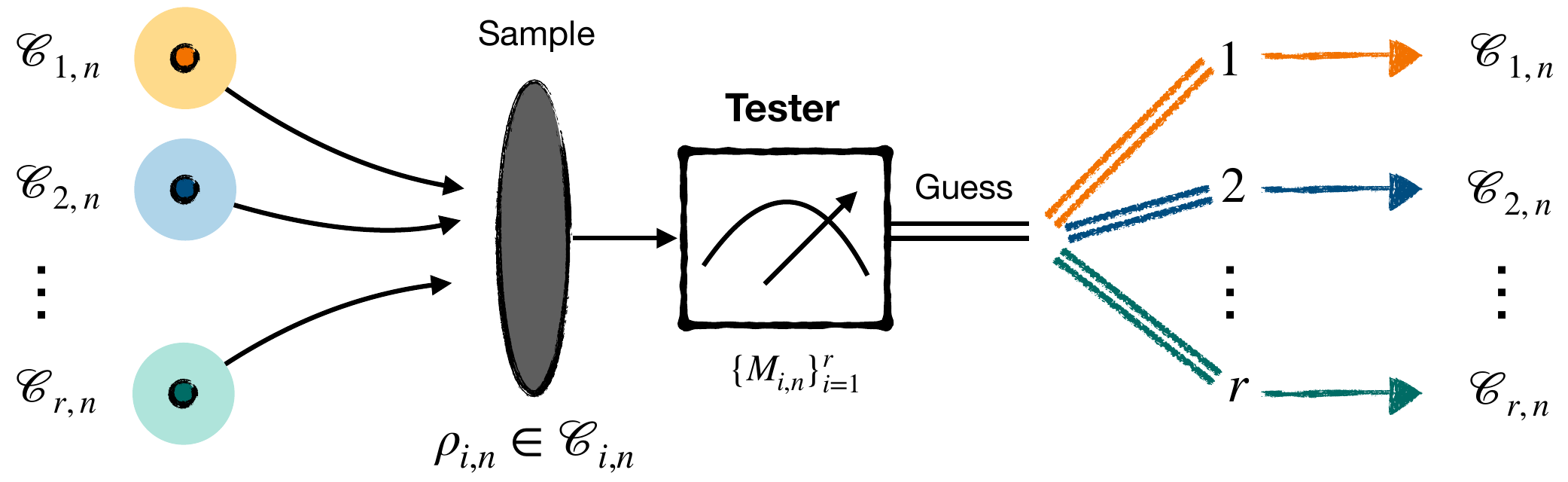}
    \caption{Quantum hypothesis testing for multiple sets of quantum states $\sC_{1,n}, \cdots, \sC_{r,n}$.}
    \label{fig: hypothesis testing demo}
\end{figure}

More precisely, consider the problem of discriminating among $r$ sets of quantum states, denoted by $\sC_{i,n}$ for $i \in \{1, \ldots, r\}$, where each set is associated with a prior probability $\pi_i$ such that $\sum_{i=1}^r \pi_i = 1$. The goal is to determine, via quantum measurement, from which set a given quantum state is drawn, without knowledge of the specific state within each set.
Let $\{M_{i,n}\}_{i=1}^r$ be a POVM for any given $n\in\NN$, where each $M_{i,n}$ corresponds to the decision that the state is drawn from set $\sC_{i,n}$. For each $i$, $\tr[\rho_{i,n} (I - M_{i,n})]$ represents the probability of incorrectly rejecting set $\sC_{i,n}$ when the true state is $\rho_{i,n} \in \sC_{i,n}$. Since the specific state within each set is unknown, we adopt a worst-case approach and define the error probability as the supremum, over all possible choices of states from each set, of the average probability of incorrectly identifying the set:
\begin{align}
    P_{e}(\{M_{i,n}\}_{i=1}^r, \{\sC_{i,n}\}_{i=1}^r) := \sup \left\{\sum_{i=1}^r \pi_i \tr[\rho_{i,n} (I -M_{i,n})]:\rho_{i,n} \in \sC_{i,n}, \forall i \right\}.
\end{align}

To characterize the fundamental limit of discrimination, we minimize the worst-case error probability over all possible POVMs:
\begin{align}
    P_{e,\min}(\{\sC_{i,n}\}_{i=1}^r) := \inf \Big\{P_{e}(\{M_{i,n}\}_{i=1}^r, \{\sC_{i,n}\}_{i=1}^r): \{M_{i,n}\}_{i=1}^r \in \POVM\,\Big\}.
\end{align}
This quantity characterizes the optimal error probability for discriminating among multiple sets of quantum states, accounting for the worst-case selection of states from each set. We are interested in the asymptotic behavior of the minimum error probability:
\begin{align}
    \lim_{n\to \infty} - \frac{1}{n} \log P_{e,\min}(\{\sC_{i,n}\}_{i=1}^r) = \ \ ?
\end{align}

Discriminating between composite and correlated hypotheses (i.e., sets of quantum states) presents several significant challenges. First, discrimination strategies must be state-agnostic, ensuring that the error probability is uniformly controlled for every possible state within each set, regardless of which specific state is sampled. Second, the optimization problem naturally takes a minimax form, requiring a simultaneous maximization over all possible states in the sets and minimization over all quantum measurements. Third, the absence of an i.i.d. structure complicates the asymptotic analysis, as standard techniques based on tensor product structures are no longer directly applicable. Fourth, it becomes necessary to define suitable extensions of the quantum Chernoff divergence to sets of quantum states that both recover known results for i.i.d. sources and accurately capture the essential features of the general composite correlated setting.

Existing results do not directly address this general scenario. While recent progress has been made in the asymmetric (Stein's) regime for binary hypotheses, where the type-II error is minimized subject to a constraint on the type-I error~\cite{fang2024generalized}, the symmetric hypothesis testing for multiple sets of quantum states remains open. 
In this work, we resolve this gap by establishing generalized quantum Chernoff bounds for the discrimination of multiple sets of quantum states, thereby extending the quantum Chernoff bound to the general composite correlated settings.

\paragraph{Summary of main results.}
\begin{itemize}
    \item (\emph{Chernoff bounds:}) We establish generalized quantum Chernoff bounds for the discrimination of multiple sets of quantum states. Specifically, Theorem~\ref{thm: chernoff bound for two sets} (for binary hypotheses) and Theorem~\ref{thm: chernoff bound for multiple sets} (for multiple hypotheses) demonstrate that, under the assumptions of convexity and stability under tensor product, the optimal error exponent is exactly characterized by the regularized Chernoff divergence between the sets. These results unify and extend the classical and quantum Chernoff bounds to the general setting of composite and correlated hypotheses, encompassing several previously known results as special cases.
    \item (\emph{Optimal tests:}) In Theorem~\ref{thm: Pemin minimax multiple}, we show that discriminating between sets of quantum states is no harder than discriminating between their most difficult elements, in terms of error probability. This minimax characterization guarantees the existence of an optimal state-agnostic test that achieves the minimum error probability for all states in the sets, matching the performance of the optimal state-dependent test for the hardest pair of states. Furthermore, Theorem~\ref{thm: optimal test for hypothesis testing} provides an explicit construction of the optimal state-agnostic test via the Holevo-Helstrom test in the binary composite hypothesis setting.
    \item (\emph{Operational interpretation:}) Theorem~\ref{thm: maximum overlap with free states} provides an operational interpretation of the maximum overlap between a pure state and a set of free states, a quantity that frequently arises in quantum resource theories, as the optimal error exponent in symmetric hypothesis testing, justifying the overlap as a meaningful measure of distinguishability in this context.
\end{itemize}

\paragraph{Organization of the paper.} The remainder of this paper is organized as follows. Section~\ref{sec: preliminaries} introduces notation, minimax theorems, and useful lemmas. Section~\ref{sec: Quantum Chernoff divergence between sets of quantum states} defines the quantum Chernoff divergence between sets of quantum states, establishes its properties, and discusses computability and nonadditivity. Section~\ref{sec: Optimal average error probability and optimal test} presents a minimax characterization of the optimal error probability and constructs the optimal test for binary composite hypotheses. Section~\ref{sec: Quantum Chernoff bound for two sets of quantum states} proves the quantum Chernoff bound for two sets of quantum states, and Section~\ref{sec: Quantum Chernoff bound for multiple sets of quantum states} extends the result to multiple sets. Section~\ref{sec: maximum overlap with free states} provides an operational interpretation of the maximum overlap with free states in resource theories. Section~\ref{sec: discussions} concludes with a discussion of open questions and future directions.

\section{Preliminaries}
\label{sec: preliminaries}

\subsection{Notations}

We adopt the following notational conventions throughout this work. Finite-dimensional Hilbert spaces are denoted by $\cH$, with $|\cH|$ representing their dimension. The set of all linear operators on $\cH$ is denoted by $\sL(\cH)$, while $\HERM(\cH)$ and $\HERM_{\pl}(\cH)$ denote the sets of Hermitian and positive semidefinite operators on $\cH$, respectively. The set of density matrices (i.e., positive semidefinite operators with unit trace) on $\cH$ is denoted by $\density(\cH)$. Calligraphic letters such as $\sA$, $\sB$, and $\sC$ are used to denote sets or sequences of sets of linear operators. Unless otherwise specified, all logarithms are taken to base two and denoted by $\log(x)$.
The positive semidefinite ordering is written as $X \geq Y$ if and only if $X - Y \geq 0$. The absolute value of an operator $X$ is defined as $|X| := (X^\dagger X)^{1/2}$. For a Hermitian operator $X$ with spectral decomposition $X = \sum_i x_i E_i$, the projection onto the non-negative eigenspaces is denoted by $\{X \geq 0\} := \sum_{x_i \geq 0} E_{i}$. The Petz \Renyi divergence is define by $D_{\Petz, \alpha}(\rho_1\|\rho_2):=\frac{1}{\alpha-1} \log Q_{\alpha}(\rho_1\|\rho_2)$ and its extension to two sets of quantum states $\sC_1$ and $\sC_2$ is defined as $D_{\Petz, \alpha}(\sC_1\|\sC_2):= \inf_{\rho_1 \in \sC_1, \rho_2 \in \sC_2} D_{\Petz, \alpha}(\rho_1\|\rho_2)$.

\subsection{Minimax theorems}

Consider a function $f: X \times Y \to \RR$ where $X,Y \subseteq \sL$ are nonempty subsets of linear operators, respectively. We always have the \emph{minimax inequalty},
\begin{align}
    \sup_{y \in Y} \inf_{x \in X} f(x,y) \leq \inf_{x \in X} \sup_{y \in Y} f(x,y).
\end{align}
If the equality holds, we call it \emph{minimax equality} and the value on both sides are \emph{minimax value}. The minimax equality is a very important property in optimization and game theory. 

\begin{definition}\label{def: saddle point minimax}
    A pair of solutions $x^* \in X$ and $y^* \in Y$ is called a saddle point of $f$ if 
    \begin{align}
        f(x^*, y) \leq f(x^*, y^*) \leq f(x, y^*), \quad \forall x \in X, y \in Y.
    \end{align}
\end{definition}

\begin{remark}\label{rem: saddle point minimax}
    From the definition, it is clear that $(x^*, y^*)$ is a saddle point of $f$ if and only if $x^* \in X$, $y^* \in Y$, and 
\begin{align}\label{eq: saddle point minimax remark}
    \sup_{y \in Y} f(x^*, y) = f(x^*, y^*) = \inf_{x \in X} f(x, y^*).
\end{align}
That is, $x^*$ minimizes against $y^*$ and $y^*$ maximizes against $x^*$.
\end{remark}

We also recall a standard characterization of saddle points from convex optimization; see, for example,~\cite[Section 3.4]{bertsekas2009convex}. We restate it here and present the proof for completeness.

\begin{lemma}\label{lem: saddle point minimax construction}
    A pair of solutions $(x^*, y^*)$ is a saddle point of $f$ if and only if the minimax equality holds and $x^*$ is an optimal solution of the problem 
    \begin{align}\label{eq: saddle point minimax tmp1}
        \min_{x \in X} \left(\sup_{y \in Y} f(x,y)\right),
    \end{align}
    while $y^*$ is an optimal solution of the problem 
    \begin{align}\label{eq: saddle point minimax tmp2}
        \max_{y \in Y} \left(\inf_{x \in X} f(x,y)\right).
    \end{align}
\end{lemma}
\begin{proof}
    Suppose that $x^*$ is an optimal solution of the problem~\eqref{eq: saddle point minimax tmp1} and $y^*$ is an optimal solution of the problem~\eqref{eq: saddle point minimax tmp2}. Then we have
    \begin{align}\label{eq: saddle point minimax tmp3}
     \sup_{y \in Y} \inf_{x \in X} f(x,y) = \inf_{x \in X} f(x, y^*)\leq f(x^*,y^*) \leq \sup_{y \in Y} f(x^*,y) = \inf_{x \in X} \sup_{y \in Y} f(x,y),
    \end{align}
    where the two equalities follow from the optimality of $x^*$ and $y^*$. Therefore, if the minimax equality holds, then equality holds throughtout above, so that 
    \begin{align}
        \sup_{y \in Y} f(x^*, y) = f(x^*, y^*) = \inf_{x \in X} f(x, y^*).
    \end{align}
    From Remark~\ref{rem: saddle point minimax}, we know that $(x^*, y^*)$ is a saddle point of $f$.

    Conversely, if $(x^*, y^*)$ is a saddle point of $f$, then we have from Eq.~\eqref{eq: saddle point minimax remark} that
    \begin{align}\label{eq: saddle point minimax tmp4}
       \inf_{x \in X} \sup_{y \in Y} f(x,y) \leq  \sup_{y \in Y} f(x^*, y) = f(x^*, y^*) = \inf_{x \in X} f(x, y^*) \leq \sup_{y \in Y} \inf_{x \in X} f(x,y).
    \end{align}
    Combined with the minimax inequality, we conclude the minimax equality. Therefore, equality holds throughtout above, so that $x^*$ is an optimal solution of the problem~\eqref{eq: saddle point minimax tmp1} and $y^*$ is an optimal solution of the problem~\eqref{eq: saddle point minimax tmp2}.
\end{proof}

The following lemma is a minimax theorem that account for the infinity values of the function~\cite[Theorem 5.2]{farkas2006potential}.
Let $X$ be a convex set in a linear space. A function $f: X \to (-\infty,-\infty]$ is convex, if $f(px+(1-p)y) \leq pf(x)+(1-p)f(y)$, the multiplication $0\cdot f(x)$ is interpreted as $0$ and $p\cdot +\infty=+\infty$ for  $p\neq 0$. Similar  definiton holds for concave functions.

\begin{lemma}\label{lem: minimax infinity value}
    Let $X$ be a compact, convex subset of a Hausdorff topological vector space and $Y$ be a convex subset of the linear space. Let $f : X \times Y \to (-\infty,+\infty]$ be lower semicontinuous on $X$ for fixed $y \in Y$, and $f$ is convex in the first and concave in the second variable. Then
\begin{align}
\sup_{y \in Y} \inf_{x \in X} f(x,y) 
= \inf_{x \in X} \sup_{y \in Y} f(x,y).
\end{align}
\end{lemma}

Another minimax theorem is given by~\cite[Corollary A.2]{mosonyi2011quantum} or~\cite[Lemma II.2]{mosonyi2023some}.

\begin{lemma}\label{lem: minimax continuous monotone}
Let $X$ be a compact topological space, $Y$ be an ordered set, and let $f : X \times Y \to \mathbb{R} \cup \{\pm \infty\}$ be a function. Assume that (i) $f(\cdot, y)$ is upper semicontinuous for every $y \in Y$, and (ii) $f(x, \cdot)$ is monotonic increasing (or decreasing) for every $x \in X$. Then
\begin{align}
\sup_{x \in X} \inf_{y \in Y} f(x, y) = \inf_{y \in Y} \sup_{x \in X} f(x, y), \tag{II.5}
\end{align}
and the suprema can be replaced by maxima.
\end{lemma}

\subsection{Useful lemmas}

The following lemmas are standard results in mathematical analysis and will be used frequently in our proofs. For detailed proofs, see, e.g.,~\cite[Lemma 2.8, 2.9]{ben2023new}.

\begin{lemma}
\label{lem: compact lsc}
Let $X$ be a nonempty compact topological space, and let $f : X \to \overline{\mathbb{R}}$ be a function. Then if $f$ is upper semicontinuous, it attains its maximum, meaning there is some $x \in X$ such that for all $x' \in X$, $f(x') \leq f(x)$. Similarly, if $f$ is lower semicontinuous, it attains its minimum.
\end{lemma}

\begin{lemma} \label{lem: inf usc}
    Let $X$ be a topological space, let $I$ be a set, and let $\{f_i\}_{i \in I}$ be a collection of functions $f_i : X \to \overline{\mathbb{R}}$. Then if each $f_i$ is upper semicontinuous, the function $f(x) = \inf_{i \in I} f_i(x)$ is also upper semicontinuous. Similarly, if each $f_i$ is lower semicontinuous, the pointwise supremum is lower semicontinuous.
\end{lemma}

The proof of the quantum Chernoff bound relies on two key lemmas, which will also play a central role in our analysis.

\begin{lemma}\label{lem: Pemin trace norm}
For any $V,W\in \PSD$, it holds that~\cite{holevo1972analog,helstrom1969quantum}
\begin{align}
    \inf_{0\leq M \leq I} \tr [(I-M)V] + \tr[MW]  = \frac{1}{2} (\tr[V+W] - \|V-W\|_1).
\end{align} 
\end{lemma}

\begin{lemma}\label{lem: audenaert}
Let $V,W \in \PSD$ and $\alpha \in [0,1]$. It holds that~\cite{audenaert2007discriminating}
\begin{align}
    \tr[V^\alpha W^{1-\alpha}] \geq \frac{1}{2} \tr[V + W - |V- W|].
\end{align}
\end{lemma}

\section{Quantum Chernoff divergence between sets of quantum states}
\label{sec: Quantum Chernoff divergence between sets of quantum states}

In this section, we introduce the quantum Chernoff divergence between sets of quantum states, which generalizes the concept from individual states to sets and serves as a fundamental measure of their distinguishability. For convex sets, this divergence admits a minimax characterization as shown in Theorem~\ref{thm: C and Q between two sets} in terms of the Chernoff quasi-divergence, and this relation extends to the asymptotic case for stable sequences of sets. Notably, the Chernoff divergence between sets can exhibit strict subadditivity, as demonstrated by an explicit example in Example~\ref{ex: nonadditivity}, which underscores the importance of regularization in its definition for sequences of sets.

\begin{definition}[Quantum Chernoff divergence between two sets of quantum states.]\label{def: chernoff divergence}
Let $\cH$ be a finite-dimensional Hilbert space, and let $\sC_1, \sC_2 \subseteq \density(\cH)$ be two sets of quantum states. The quantum Chernoff divergence between these sets is defined as
\begin{align}
    C(\sC_1\|\sC_2) := \inf_{\substack{\rho_1 \in \sC_1\\ \rho_2 \in \sC_2}} C(\rho_1\|\rho_2).
\end{align}

Moreover, let ${\sC}_1 = \{\sC_{1,n}\}_{n\in \NN}$ and $\sC_2 = \{\sC_{2,n}\}_{n\in \NN}$ be two sequences of sets of quantum states, where each $\sC_{1,n}, \sC_{2,n} \subseteq \density(\cH^{\ox n})$. The regularized quantum Chernoff divergence between these sequences is defined as
\begin{align}
    \underline{C}^\reg(\sC_1\|\sC_2) & := \liminf_{n\to \infty} \frac{1}{n} C(\sC_{1,n}\| \sC_{2,n}),\\
    \overline{C}^\reg(\sC_1\|\sC_2) & := \limsup_{n\to \infty} \frac{1}{n} C(\sC_{1,n}\| \sC_{2,n}).
\end{align}
If the following limit exists, we define the regularized Chernoff divergence as
\begin{align}
    C^\reg(\sC_1\|\sC_2) := \lim_{n\to \infty} \frac{1}{n} C(\sC_{1,n}\| \sC_{2,n}).
\end{align}
\end{definition}

\begin{remark}[Attainment.]
    The quantum Chernoff divergence can be written in terms of the Petz \Renyi divergence as $C(\rho_1\|\rho_2) = \sup_{\alpha \in (0,1)} (1-\alpha) D_{\Petz,\alpha}(\rho_1\|\rho_2)$. For any fixed $\alpha$, $D_{\Petz,\alpha}(\rho_1\|\rho_2)$ is lower semicontinuous in $(\rho_1,\rho_2) \in \PD \times \PD$~\cite[Proposition III.11]{mosonyi2023some}. Since $C(\rho_1\|\rho_2)$ is the pointwise supremum of lower semicontinuous functions, it is itself lower semicontinuous in $(\rho_1,\rho_2)$ by Lemma~\ref{lem: inf usc}. Therefore, if the sets $\sC_1$ and $\sC_2$ are compact, Lemma~\ref{lem: compact lsc} implies that the infimum in the definition of $C(\sC_1\|\sC_2)$ is attained.
\end{remark}

In many application scenarios, the sequences of sets under consideration are not arbitrary but possess a structure that is compatible with tensor products. This property, known as \emph{stability} (or closeness) under tensor product, is formalized as follows.

\begin{boxdefinition}[Stable sequence.]\label{def: closed under tensor product}
Let $\sC_1 \subseteq \PSD(\cH_1)$, $\sC_2 \subseteq \PSD(\cH_2)$, and $\sC_{3} \subseteq \PSD(\cH_1 \otimes \cH_2)$. We say that $\{\sC_1, \sC_2, \sC_{3}\}$ is stable under tensor product if, for any $X_1 \in \sC_1$ and $X_2 \in \sC_2$, it holds that $X_1 \otimes X_2 \in \sC_{3}$. In short, we write $\sC_1 \otimes \sC_2 \subseteq \sC_{3}$. A sequence of sets $\{\sC_n\}_{n \in \NN}$ with $\sC_n \subseteq \PSD(\cH^{\ox n})$ is called stable under tensor product if $\sC_n \otimes \sC_m \subseteq \sC_{n+m}$ for all $n, m \in \NN$.
\end{boxdefinition}

\begin{remark}[Subadditivity.]
By its operational meaning, the quantum Chernoff divergence is additive under $n$-fold tensor product states, i.e., for any $n \in \NN$, $C(\rho_1^{\otimes n}\|\rho_2^{\otimes n}) = n\, C(\rho_1\|\rho_2)$. More generally, the Chernoff divergence is subadditive under tensor products of different states:
\begin{align}\label{eq: subadditivity of Chernoff divergence}
    C(\rho_1 \otimes \sigma_1\|\rho_2 \otimes \sigma_2) \leq C(\rho_1\|\rho_2) + C(\sigma_1\|\sigma_2).
\end{align}
This can be seen as follows:
\begin{align}
    C(\rho_1 \otimes \sigma_1\|\rho_2 \otimes \sigma_2) & = \max_{\alpha \in [0,1]}  - \log Q_\alpha (\rho_1 \otimes \sigma_1\|\rho_2 \otimes \sigma_2)  \\
    & = \max_{\alpha \in [0,1]}  - \log Q_\alpha (\rho_1\|\rho_2) - \log Q_\alpha(\sigma_1\|\sigma_2)  \\
    & \leq \max_{\alpha \in [0,1]}  - \log Q_\alpha (\rho_1\|\rho_2)  + \max_{\alpha \in [0,1]}  - \log Q_\alpha(\sigma_1\|\sigma_2)  \\
    & = C(\rho_1\|\rho_2) + C(\sigma_1\|\sigma_2),
\end{align}
where the second equality uses the multiplicativity of $Q_\alpha$ under tensor products, and the inequality follows from splitting the maximization over $\alpha$ for each term.
\end{remark}

\begin{remark}[Limit existence.]\label{rem: limit existence for Chernoff}
As a consequence of the subadditivity of the Chernoff divergence in Eq.~\eqref{eq: subadditivity of Chernoff divergence}, its extension to sets of quantum states is also subadditive, provided that the sequences $\sC_1 = \{\sC_{1,n}\}_{n\in \NN}$ and $\sC_2 = \{\sC_{2,n}\}_{n\in \NN}$ are both stable under tensor product~\cite[Lemma 26]{fang2024generalized}. Therefore, the regularized quantum Chernoff divergence exists by Fekete's lemma and satisfies
\begin{align}\label{eq: multiplicativity of Chernoff divergence}
    C^\reg(\sC_1\|\sC_2) = \overline{C}^\reg(\sC_1\|\sC_2) = \underline{C}^\reg(\sC_1\|\sC_2) = \inf_{n \geq 1} \frac{1}{n} C(\sC_{1,n}\| \sC_{2,n}).
    \end{align}
    At the end of this section, we provide an explicit example in Example~\ref{ex: nonadditivity} demonstrating that the quantum Chernoff divergence between two sets can be strictly subadditive, thereby illustrating the necessity of regularization in general.
\end{remark}

\bigskip
Analogous to the extension of the quantum Chernoff divergence to sets of quantum states, we also generalize the quantum Chernoff quasi-divergence.

\begin{definition}[Quantum Chernoff quasi-divergence between two sets of quantum states.]\label{def: chernoff quasi divergence}  
Let $\alpha \in [0,1]$ and $\cH$ be a finite-dimensional Hilbert space. Let $\sC_1, \sC_2 \subseteq \density(\cH)$ be two sets of quantum states. The quantum Chernoff quasi-divergence between these sets is defined as 
\begin{align}
Q_\alpha(\sC_1\|\sC_2) := \sup_{\substack{\rho_1 \in \sC_1\\ \rho_2 \in \sC_2}} Q_\alpha(\rho_1\|\rho_2).
\end{align}

Moreover, let $\sC_1 = \{\sC_{1,n}\}_{n\in \NN}$ and $\sC_2 = \{\sC_{2,n}\}_{n\in \NN}$ be two sequences of sets of quantum states, where each $\sC_{1,n}, \sC_{2,n} \subseteq \density(\cH^{\ox n})$. The regularized quantum Chernoff quasi-divergence between these sequences is defined as
\begin{align}
\underline{Q}_\alpha^\reg(\sC_1\|\sC_2) & := \liminf_{n\to \infty}  [Q_\alpha(\sC_{1,n}\| \sC_{2,n})]^\frac{1}{n},\\
\overline{Q}_\alpha^\reg(\sC_1\|\sC_2) & := \limsup_{n\to \infty} [Q_\alpha(\sC_{1,n}\| \sC_{2,n})]^\frac{1}{n}.
\end{align}
If the following limit exists, we define the regularized Chernoff quasi-divergence as
\begin{align}
    Q_{\alpha}^\reg(\sC_1\|\sC_2) := \lim_{n\to \infty}  [Q_{\alpha}(\sC_{1,n}\| \sC_{2,n})]^{\frac{1}{n}}.
\end{align}
\end{definition}
\begin{remark}[Limit existence.]\label{rem: limit existence for Chernoff quasi}
The multiplicativity of the quantum Chernoff quasi-divergence under tensor product implies that its extension to sets is supermultiplicative, provided that the sequences are stable under tensor product. Consequently, the regularized quasi-divergence exists by Fekete's lemma and satisfies
\begin{align}\label{eq: multiplicativity of Chernoff quasi-divergence}
    Q_\alpha^\reg(\sC_1\|\sC_2) = \overline{Q}_\alpha^\reg(\sC_1\|\sC_2) = \underline{Q}_\alpha^\reg(\sC_1\|\sC_2) = \sup_{n \geq 1}  [Q_{\alpha}(\sC_{1,n}\| \sC_{2,n})]^{\frac{1}{n}}.
\end{align}
\end{remark}

While the quantum Chernoff divergence and quasi-divergence between two states are directly related through their definitions, we have thus far extended these quantities for sets of quantum states independently. It is therefore natural to ask whether a relationship analogous to $C(\rho_1 \| \rho_2) = \max_{0 \leq \alpha \leq 1}  -\log Q_\alpha(\rho_1 \| \rho_2)$ holds between the Chernoff divergence and quasi-divergence when extended to sets of quantum states. The following result establishes this connection by applying minimax theorems for convex sets.

\begin{boxtheorem}[Finite and asymptotic connections.]\label{thm: C and Q between two sets}
Let $\cH$ be a finite-dimensional Hilbert space, and let $\sC_1, \sC_2 \subseteq \density(\cH)$ be two convex sets of quantum states. Then it holds that
    \begin{align}\label{eq: C and Q between two sets}
    C(\sC_{1}\|\sC_{2}) = \max_{\alpha \in [0,1]} -\log Q_\alpha(\sC_{1}\|\sC_{2}).
\end{align} 
Moreover, let $\sC_1 = \{\sC_{1,n}\}_{n\in \NN}$ and $\sC_2 = \{\sC_{2,n}\}_{n\in \NN}$ be two stable sequences of convex sets of quantum states, where each $\sC_{1,n}, \sC_{2,n} \subseteq \density(\cH^{\ox n})$. Then it holds that
\begin{align}
    C^\infty(\sC_1\|\sC_2) = \max_{\alpha \in [0,1]} -\log Q_\alpha^{\infty}(\sC_1\|\sC_2).
\end{align}
\end{boxtheorem}

\begin{proof}
We have the following chain of equalities:
\begin{align}
    C(\sC_{1}\|\sC_{2}) & = \inf_{\substack{\rho_1 \in \sC_{1}\\\rho_2 \in \sC_{2}}} C(\rho_1\|\rho_2)\\
    & = \inf_{\substack{\rho_1 \in \sC_{1}\\\rho_2 \in \sC_{2}}} \max_{0\leq \alpha \leq 1} -\log Q_\alpha(\rho_1\|\rho_2)\\
    & = -\log  \sup_{\substack{\rho_1 \in \sC_{1}\\\rho_2 \in \sC_{2}}} \min_{0\leq \alpha \leq 1} Q_\alpha(\rho_1\|\rho_2)\\
    & = -\log  \min_{0\leq \alpha\leq 1} \sup_{\substack{\rho_1 \in \sC_{1}\\\rho_2 \in \sC_{2}}}  Q_\alpha(\rho_1\|\rho_2)\\
    & =  \max_{0\leq \alpha\leq 1} \inf_{\substack{\rho_1 \in \sC_{1}\\\rho_2 \in \sC_{2}}}  -\log Q_\alpha(\rho_1\|\rho_2)\\
    & =  \max_{0\leq \alpha\leq 1}  -\log Q_\alpha(\sC_{1}\|\sC_{2}),\label{eq: C set minimax tmp1}
\end{align}
where the first two equalities follow directly from the definitions. By Lieb's concavity theorem, $Q_\alpha(\rho_1\|\rho_2)$ is jointly concave in $(\rho_1, \rho_2)$ (see also~\cite[Proposition 4.8]{Tomamichel2015b}), and it is convex and continuous in $\alpha$~\cite{audenaert2008asymptotic}. These properties allow us to apply the minimax theorem in Lemma~\ref{lem: minimax infinity value} to exchange the order of the supremum and infimum. The final equality then follows from the definition of $Q_\alpha(\sC_{1}\|\sC_{2})$.

We now prove the second statement. Suppose $\sC_1 = \{\sC_{1,n}\}_{n\in \NN}$ and $\sC_2 = \{\sC_{2,n}\}_{n\in \NN}$ are stable sequences under tensor product. By Eq.~\eqref{eq: multiplicativity of Chernoff divergence}, the regularized Chernoff divergence exists and the following chain of equalities holds:
\begin{align}
    C^\reg(\sC_1\|\sC_2) & = \lim_{n\to \infty}  \frac{1}{n} C(\sC_{1,n}\| \sC_{2,n})\\
    & =  \lim_{n\to \infty}  \frac{1}{2^n} C(\sC_{1,2^n}\|\sC_{2,2^n})\\
    & = \inf_{n \geq 1} \frac{1}{2^n} C(\sC_{1,2^n}\|\sC_{2,2^n})\\
    & = \inf_{n \geq 1} \frac{1}{2^n} \max_{\alpha \in [0,1]}  -\log Q_\alpha(\sC_{1,2^n}\|\sC_{2,2^n}) \\
    & = \inf_{n \geq 1}  \max_{\alpha \in [0,1]}  - \frac{1}{2^n} \log Q_\alpha(\sC_{1,2^n}\|\sC_{2,2^n}) ,
\end{align}
where the second line follows by restricting to the subsequence, the third line follows from the subadditivity of the sequence and Fekete's lemma, the fourth line follows from Eq.~\eqref{eq: C set minimax tmp1} established above. Define 
\begin{align}
f(n,\alpha):= - \frac{1}{2^n} \log Q_\alpha(\sC_{1,2^n}\|\sC_{2,2^n}).
\end{align} 
Due to the supermultiplicativity of $Q_\alpha(\sC_{1,2^n}\|\sC_{2,2^n})$, we have 
\begin{align}
f(n+1,\alpha) & = -\frac{1}{2^{n+1}} \log Q_{\alpha}(\sC_{1,2^{n+1}}\|\sC_{2,2^{n+1}})\\
& \leq -\frac{1}{2^{n+1}} \log \left( Q_{\alpha}(\sC_{1,2^n}\|\sC_{2,2^n}) \right)^2  = f(n,\alpha).
\end{align} 
So the objective function $f(n,\alpha)$ is monotone decreasing in $n$ for each fixed $\alpha$. Furthermore, since $Q_{\alpha}(\rho_n\|\sigma_n)$ is continuous in $\alpha \in [0,1]$, and $Q_{\alpha}(\sC_{1,n}\| \sC_{2,n})$ is defined as the pointwise supremum over lower semicontinuous functions, Lemma~\ref{lem: inf usc} ensures that $Q_{\alpha}(\sC_{1,n}\| \sC_{2,n})$ is also lower semicontinuous in $\alpha$. Consequently, $f(n,\alpha)$ is upper semicontinuous in $\alpha$ for each $n$. By applying the minimax theorem in Lemma~\ref{lem: minimax continuous monotone}, we can obtain
\begin{align}
    C^\reg(\sC_1\|\sC_2) & =  \max_{\alpha \in [0,1]} \inf_{n \geq 1}  - \frac{1}{2^n} \log Q_\alpha(\sC_{1,2^n}\|\sC_{2,2^n}) \\
    & = \max_{\alpha \in [0,1]}  -    \log  \sup_{n \geq 1} [Q_\alpha(\sC_{1,2^n}\|\sC_{2,2^n})]^{\frac{1}{2^n}} \\
    & = \max_{\alpha \in [0,1]}  -\log \lim_{n\to \infty} [Q_\alpha(\sC_{1,2^n}\|\sC_{2,2^n})]^{\frac{1}{2^n}} \\
    & = \max_{\alpha \in [0,1]}  -\log Q_\alpha^\reg(\sC_1\|\sC_2) ,
\end{align}
where the last two lines follow from the supermultiplicativity of the sequence and Fekete's lemma. This completes the proof.
\end{proof}

\begin{remark}[Computability.]\label{rem: computability of Chernoff sets}
For any fixed $\alpha \in [0,1]$, the function $Q_\alpha(\rho\|\sigma)$ is jointly concave in $(\rho,\sigma)$, so the quasi-divergence $Q_{\alpha}(\sC_1\|\sC_2)$ can be efficiently computed using the QICS package~\cite{he2024qics} whenever $\sC_1$ and $\sC_2$ admit semidefinite representations. If the sets $\sC_1$ and $\sC_2$ possess additional symmetries, the computational complexity can be further reduced. 

Moreover, since $Q_{\alpha}(\sC_1\|\sC_2)$ is convex and lower semicontinuous in $\alpha$, the minimum over $\alpha \in [0,1]$ is achieved at a unique optimal solution. Thus, the quantum Chernoff divergence $C(\sC_1\|\sC_2)$ can also be efficiently computed by scanning over $\alpha$. 

On the other hand, since $C(\rho\|\sigma)=0$ if and only if $\rho = \sigma$, $C(\sC_1\|\sC_2)$ contains the separability problem~\cite{gurvits2003classical} as a special instance. Therefore, computing the quantum Chernoff divergence between sets of states can be hard in general if the sets do not have semidefinite representations.
\end{remark}

Using the above discussions, we now provide an explicit example for the nonadditivity of quantum Chernoff divergence between sets of quantum states.  

\begin{boxexample}[Nonadditivity of Chernoff divergence between sets of quantum states.]\label{ex: nonadditivity}
Consider two qutrit quantum channels. Let $\cN(\cdot)=\tr[\cdot] \rho$ to be the replacer channel with 
\begin{align}
    \rho = 0.9 \cdot \ket{\psi}\bra{\psi} + 0.1 \cdot \frac{I}{3}, \quad \text{where}\quad \ket{\psi} = \frac{1}{\sqrt{2}}(\ket{0}+\ket{2}).
\end{align}
Let $\cM$ be the platypus channel~\cite{Leditzky_2023}, $\cM(X) = M_0 X M_0^\dagger + M_1 X M_1^\dagger$ with Kraus operators
\begin{align}
    M_0 = \begin{bmatrix}
        \sqrt{p} & 0 & 0\\
        0 & 0 & 0\\
        0 & 1 & 0
    \end{bmatrix},
    \qquad 
    M_1 = \begin{bmatrix}
        0 & 0 & 0\\
        \sqrt{1-p} & 0 & 0\\
        0 & 0 & 1
    \end{bmatrix}.
\end{align}
In this case, we have the image sets of these channels as
\begin{align}
    \sC_{1,n} & = \{\rho^{\ox n}\}, \quad \text{and} \quad
    \sC_{2,n} = \{\cM^{\ox n}(\sigma_n): \sigma_n \in \density((\CC^3)^{\ox n})\}.
\end{align}
As these sets are given by semidefinite constraints, we can efficiently evaluate the Petz \Renyi divergence and Chernoff divergence between them via semidefinite programming (see Remark~\ref{rem: computability of Chernoff sets}). More explicitly, for fixed $\alpha \in (0,1)$, we can efficiently evaluate 
\begin{align}
D_{\Petz, \alpha}(\sC_{1,n}\| \sC_{2,n}) & = \inf_{\sigma_n \in \density} D_{\Petz,\alpha}(\rho^{\ox n}\|\cM^{\ox n}(\sigma_n)),
\end{align}
by the QICS package~\cite{he2024qics}.
Moreover, we can also evaluate the Chernoff divergence
\begin{align}
C(\sC_{1,n}\| \sC_{2,n}) = \sup_{\alpha \in (0,1)} (1-\alpha) D_{\Petz,\alpha}(\sC_{1,n}\| \sC_{2,n}),
\end{align}
by scaning the parameter $\alpha \in (0,1)$ with fine grid.
\end{boxexample}

The numerical results are shown in Figure~\ref{fig:petz_chernoff_nonadd} where we use parameter $s=0.01$ and scan $\alpha\in(0,1)$ with step size $0.01$. Panel (a) displays the Petz \Renyi divergence $D_{\Petz,\alpha}(\sC_{1,n}\| \sC_{2,n})/n$ for $n=1,2,3$ as a function of $\alpha$, while panel (b) shows the objective function of the Chernoff divergence versus $\alpha$, together with its maximum (i.e., the Chernoff divergence $C(\sC_{1,n}\| \sC_{2,n})/n$) for $n=1,2,3$. The plots exhibit a clear separation between different numbers of copies for both the Petz \Renyi and Chernoff divergences, illustrating non-additivity in this example and justifying the necessity of regularization in our definition and the later main results.

\begin{figure}[htp]
\centering
\begin{minipage}{0.8\linewidth}
    \centering
    \includegraphics[width=\linewidth]{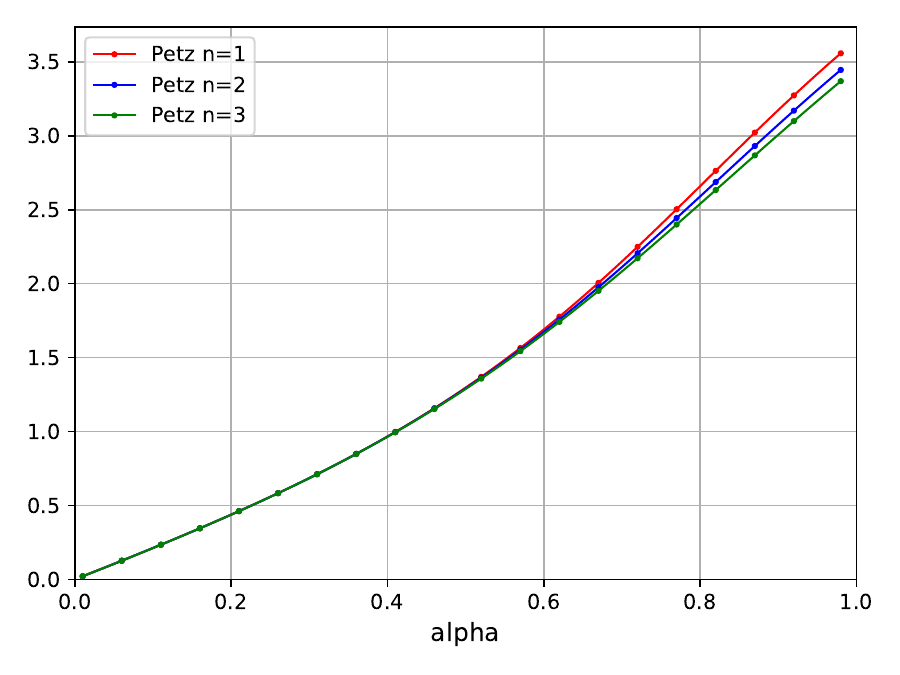}
    \caption*{(a) Petz \Renyi divergence}
\end{minipage}
\vspace{6pt}
\begin{minipage}{0.8\linewidth}
    \centering
    \includegraphics[width=\linewidth]{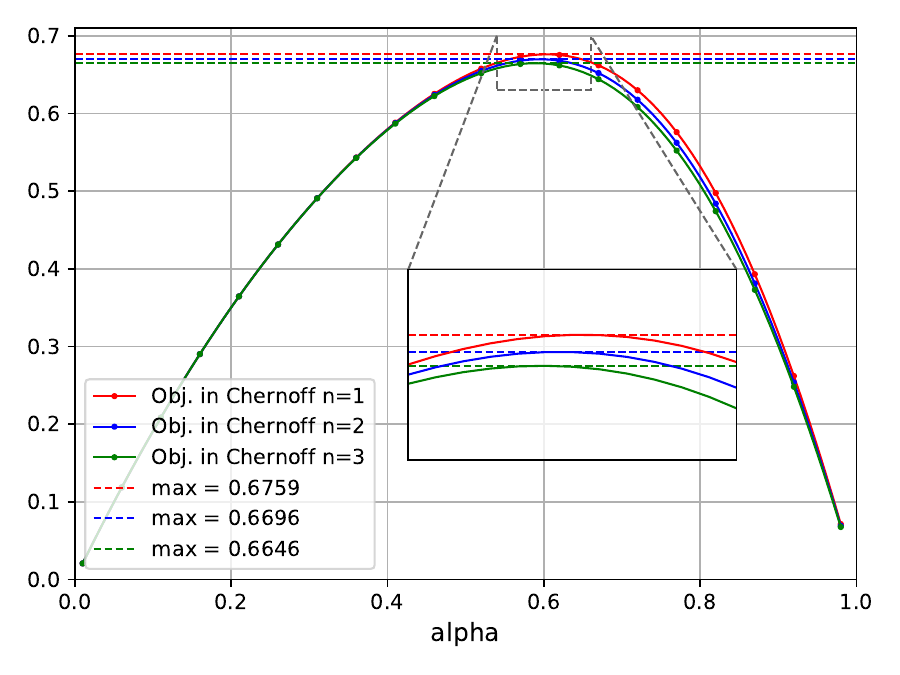}
    \caption*{(b) Chernoff divergence}
\end{minipage}
\caption{Nonadditivity for the Petz \Renyi divergence and Chernoff divergence given in Example~\ref{ex: nonadditivity}. Here, we consider the replacer channel $\cN$ with output state $\rho = 0.9 \cdot \ket{\psi}\bra{\psi} + 0.1 \cdot I/3$, where $\ket{\psi} = (\ket{0}+\ket{2})/\sqrt{2}$, and the platypus channel $\cM$ with channel parameter $s=0.01$. We plot (a) the Petz \Renyi divergence $D_{\Petz,\alpha}(\sC_{1,n}\| \sC_{2,n})/n$ and (b) the objective function in Chernoff divergence for $n=1,2,3$ as functions of $\alpha \in (0,1)$, respectively. The Chernoff divergence $C(\sC_{1,n}\| \sC_{2,n})/n$ is indicated by the maximum value in dashed lines.}
\label{fig:petz_chernoff_nonadd}
\end{figure}

\section{Optimal error probability and optimal test for sets of states}
\label{sec: Optimal average error probability and optimal test}

In this section, we present a minimax formula for the optimal average error probability when distinguishing between multiple sets of quantum states. Specifically, the minimum error probability for composite hypotheses equals the maximum, over all possible choices of states from each set, of the minimum error probability for those states. This reduces the composite problem to the hardest pairwise discrimination. For binary hypotheses, we also give an explicit construction of the optimal state-agnostic test, which is the Holevo-Helstrom test for the most difficult pair of states.

\subsection{Optimal average error probability}

The following result establishes a fundamental relation between discriminating sets of quantum states and discriminating their most challenging elements, based on the minimax theorem.

\begin{boxtheorem}[Optimal error probability.]\label{thm: Pemin minimax multiple}
Let $\cH$ be a finite-dimensional Hilbert space, and let $\{\sC_i\}_{i=1}^r$ be $r$ convex sets of quantum states, where each $\sC_i \subseteq \density(\cH)$. Then it holds that
    \begin{align}\label{eq: Pemin minimax multiple}
P_{e,\min}(\{\sC_i\}_{i=1}^r) = \sup_{\forall i,\, \rho_i\in \sC_i}  P_{e,\min}(\{\rho_i\}_{i=1}^r),
\end{align}
where the suprema can be replaced by maxima if the sets are also compact.
\end{boxtheorem}

\begin{proof}
By definition, we have 
\begin{align}
   P_{e,\min}(\{\sC_i\}_{i=1}^r) = \inf_{\{M_i\}_{i=1}^r} \sup_{\forall i,\, \rho_i \in \sC_i} \sum_{i=1}^r \pi_i \tr[\rho_i (I - M_i)]
\end{align}
Due to the linearity of the error probability $\sum_{i=1}^r \pi_i \tr[\rho_i (I - M_i)]$ in both the measurement operators $\{M_i\}_{i=1}^r$ and the states $\{\rho_i\}_{i=1}^r$, and since the set of all POVMs is convex and compact while each $\sC_i$ is convex by assumption, we can apply Lemma~\ref{lem: minimax infinity value} to exchange the infimum and supremum, yielding
\begin{align}\label{eq: Pemin minimax tmp}
P_{e,\min}(\{\sC_i\}_{i=1}^r) = \sup_{\forall i,\, \rho_i \in \sC_i}  \inf_{\{M_i\}_{i=1}^r} \sum_{i=1}^r \pi_i \tr[\rho_i (I - M_i)]
    & = \sup_{\forall i,\, \rho_i \in \sC_i}  P_{e,\min}(\{\rho_i\}_{i=1}^r),
\end{align}
where the second equality follows from Eq.~\eqref{eq: Pemin multiple states}. 
Note that $P_{e,\min}(\{\rho_i\}_{i=1}^r)$ is upper semicontinuous by Lemma~\ref{lem: inf usc}. Therefore, by Lemma~\ref{lem: compact lsc}, the supremum can be replaced by a maximum whenever the sets are compact.
\end{proof}

\begin{remark}[Computability.]
For a fixed collection of $r$ quantum states, the minimum error probability can be formulated as a semidefinite program (SDP) (see, e.g.,~\cite[Eq. (39)]{li2016discriminating}):
\begin{align}
    P_{e,\min}(\{\rho_i\}_{i=1}^r) = \max_{X\in \HERM} \left\{1-\tr X: X \geq \pi_i \rho_i,\, \forall i = 1,\ldots, r \right\}.
\end{align}
By Theorem~\ref{thm: Pemin minimax multiple}, this extends to
\begin{align}
    P_{e,\min}(\{\sC_i\}_{i=1}^r) = \sup_{\forall i,\, \rho_i \in \sC_i} \max_{X\in \HERM} \left\{1-\tr X: X \geq \pi_i \rho_i,\, \forall i = 1,\ldots, r \right\},
\end{align}
which is also an SDP whenever the sets $\sC_i$ admit semidefinite representations. In such cases, the optimal value can be efficiently computed.
\end{remark}

\subsection{Optimal test for binary composite hypotheses}
\label{sec: optimal test for binary composite hypotheses}
In quantum hypothesis testing between two quantum states $\rho_1$ and $\rho_2$ with prior probabilities $\pi_1$ and $\pi_2$, the optimal measurement is given by the Holevo-Helstrom test $\{\pi_1 \rho_1 - \pi_2 \rho_2 \geq 0\}$~\cite{holevo1972analog,helstrom1969quantum}. However, for composite hypotheses, where only partial information about the states is available, the problem becomes much more challenging. Here, one must design a test that \emph{universally} minimizes the average error probability for all possible states within the specified sets. In particular, the minimax equality in Eq.~\eqref{eq: Pemin minimax tmp} of Theorem~\ref{thm: Pemin minimax multiple},
\begin{align}
\inf_{\{M_i\}_{i=1}^r} \sup_{\forall i,\, \rho_i \in \sC_i} \sum_{i=1}^r \pi_i \tr[\rho_i (I - M_i)] = \sup_{\forall i,\, \rho_i \in \sC_i}  \inf_{\{M_i\}_{i=1}^r} \sum_{i=1}^r \pi_i \tr[\rho_i (I - M_i)],
\end{align}
guarantees the existence of such an optimal state-agnostic test: one that achieves the minimum error probability for all states in the sets (i.e., the minimizer on the left hand side), matching the performance of the optimal state-dependent test for the most difficult pair of states (i.e., the minimizer on the right hand side). While the minimax theorem ensures the existence of an optimal state-agnostic test, it does not provide an explicit construction. 

In this section, we explicitly construct an optimal test for binary composite hypotheses by analyzing the saddle point (or equilibrium point) of the minimax problem in more detail.

\subsubsection{Minimax optimization and saddle points}

To obtain saddle points, we may calculate the inner ``sup'' and ``inf'' functions appearing in the Lemma~\ref{lem: saddle point minimax construction}, then minimize and maximize them, respectively, and obtain the corresponding sets of minima $X^*$ and maxima $Y^*$. If the optimal values are equal (i.e., minimax equality holds), the set of saddle points is $X^* \times Y^*$. Otherwise, there are no saddle points~\cite[Proposition 1.4]{ekeland1999convex}.

While this standard approach to finding saddle points involves evaluating both sides of the minimax inequality, we provide an alternative method in the following. It allows for the construction of saddle points by optimizing only one side of the minimax problem, which will prove useful in the explicit construction of optimal tests later.

\begin{boxlemma}[Saddle point by one-side optimization.]\label{lem: saddle point minimax construction uniqueness}
Let $f: X \times Y \to \RR$ and $g(y) = \inf_{x \in X} f(x,y)$. Suppose $y^*\in Y$ is a maximizer of $g(y)$, i.e., $g(y^*) = \sup_{y \in Y} g(y)$, and $x^* \in X$ is a minimizer of $\inf_{x \in X} f(x,y^*)$, i.e., $f(x^*,y^*) = \inf_{x \in X} f(x,y^*)$. If the minimax equality holds for function $f$ and the optimization $\inf_{x\in X}f(x,y^*)$ has a unique minimizer, then $(x^*, y^*)$ is a saddle point. Consequently, $x^*$ is a minimizer of the optimization $\inf_{x\in X}[\sup_{y \in Y} f(x,y)]$.
\end{boxlemma}

\begin{proof}
Let $x^{**} \in X$ be a minimizer of $\inf_{x \in X} [\sup_{y\in Y} f(x,y)]$. Then we know from Lemma~\ref{lem: saddle point minimax construction} that $(x^{**}, y^*)$ is a saddle point of $f$. By Eq.~\eqref{eq: saddle point minimax tmp4} (note that all equalities holds), any saddle point $(x^*, y^*)$ gives the same minimax value. So we have 
\begin{align}
    f(x^{**}, y^*) = f(x^*, y^*) = \inf_{x \in X} f(x, y^*).
\end{align}
As we assume that $\inf_{x\in X} f(x,y^*)$ has a unique minimizer, we have $x^{**} = x^*$. Therefore, $(x^*, y^*)$ is a saddle point of $f$.
\end{proof}

We note that the uniqueness of the minimizer in $\inf_{x \in X} f(x, y^*)$ is crucial for Lemma~\ref{lem: saddle point minimax construction uniqueness} to apply. If the minimizer is not unique, it is generally unclear whether any minimizer will yield a saddle point. In particular, not every minimizer of $\inf_{x \in X} f(x, y^*)$ is necessarily a minimizer of $\inf_{x \in X}[\sup_{y \in Y} f(x, y)]$. This is illustrated by the following counterexample:
\begin{align}
    \min_{x\in [-1,1]} \max_{y\in[-1,1]} xy =  \max_{y\in[-1,1]} \min_{x\in [-1,1]} xy.
\end{align} 
If we solve $\max_{y\in[-1,1]} (\min_{x\in [-1,1]} xy)$, we find $y^* = 0$, and any $x^* \in [-1,1]$ is a minimizer of $\min_{x\in [-1,1]} x y^*$. However, only $x^* = 0$ is a minimizer of $\min_{x\in [-1,1]} (\max_{y\in[-1,1]} xy)$; that is, only $(0,0)$ is a saddle point.

\begin{boxlemma}\label{lem: optimal test uniqueness}
    Let $X \in \HERM$ be a full rank Hermitian operator. Then the optimal solution to the semidefinite program $\max_{0\leq M \leq I} \tr[XM]$ is unique and is given by $\{X \geq 0\}$.
\end{boxlemma}

\begin{proof}
    Let $X = \sum_{i=1}^d \lambda_i \ket{\psi_i}\bra{\psi_i}$ be the spectral decomposition of $X$, where $\lambda_i > 0$ for $i\in \{1,\cdots,k\}$ and $\lambda_i < 0$ for $i\in \{k+1,\cdots,d\}$. For any $M$, let $m_{ij} = \<\psi_i|M|\psi_j\>$. Then we have
    \begin{align}
        \tr[XM] = \sum_{i=1}^d \lambda_i \tr[\ket{\psi_i}\bra{\psi_i}M] = \sum_{i=1}^d \lambda_i m_{ii} = \sum_{i=1}^k \lambda_i m_{ii} + \sum_{i=k+1}^d \lambda_i m_{ii} \leq \sum_{i=1}^k \lambda_i,
    \end{align}
    where the equality holds if and only if $m_{ii} = 1$ for $i\in \{1,\cdots,k\}$ and $m_{ii} = 0$ for $i\in \{k+1,\cdots,d\}$. Denote the matrix form of $M$ in the basis of $\{\ket{\psi_i}\}_{i=1}^d$ as 
    \begin{align}
    \begin{pmatrix}
        M_{11} & M_{12}\\ M_{21} & M_{22}
    \end{pmatrix}
    \end{align}
    where $M_{11}$ is of size $k\times k$ and $M_{22}$ is of size $(d-k)\times (d-k)$. As $M \geq 0$, we have $M_{22} \geq 0$. Since the diagonal elements of $M_{22}$ are all zeros, we have $M_{22}=0$. This further implies that $M_{12} = 0$ and $M_{21} = 0$. So any optimal $M$ is in the form of block diagonal matrix 
    \begin{align}
    \begin{pmatrix}
        M_{11} & 0\\ 0 & 0
    \end{pmatrix},
    \end{align}
    where $0$ above represents the zero matrix of suitable size. Since $M \leq I$, we have $M_{11}\leq I$. Moreover, since the dignoal elements of $M_{11}$ are all ones, we can conclude that any offdiagonal elements of $M_{11}$ are zeros. This is because any principle submatrix
    \begin{align}
    \begin{pmatrix}
    1 & x \\ x^* & 1
    \end{pmatrix} \leq \begin{pmatrix} 1 & 0 \\ 0 & 1\end{pmatrix}
    \end{align}
    implies that $x = 0$. This shows that the optimal solution is unique and is given by $\sum_{i=1}^k \ket{\psi_i}\bra{\psi_i} = \{X \geq 0\}$. This concludes the proof.
\end{proof}

\subsubsection{Optimal test for binary hypotheses}

With Lemma~\ref{lem: saddle point minimax construction uniqueness} and Lemma~\ref{lem: optimal test uniqueness} in place, we are now prepared to explicitly construct the optimal test. The following result shows that, for binary composite hypotheses, the optimal state-agnostic measurement is given by the Holevo-Helstrom test corresponding to the pair of states from the two sets that achieve the minimal trace distance.

\begin{boxtheorem}[Optimal test for binary hypotheses.]\label{thm: optimal test for hypothesis testing}
 Let $\cH$ be a finite-dimensional Hilbert space. Let $\sC_1, \sC_2 \subseteq \density(\cH)$ be two convex, compact sets of quantum states and $\{\pi_1, \pi_2\}$ be the prior probabilities. Suppose $(\rho_1^*, \rho_2^*) \in \sC_1 \times \sC_2$ is a minimizer of the convex optimization problem $\min_{\rho_1 \in \sC_1,\, \rho_2 \in \sC_2} \|\pi_1 \rho_1 - \pi_2 \rho_2\|_1$. Then, the optimal state-agnostic test for discriminating between $\sC_1$ and $\sC_2$ is given by the projection $\{\pi_1 \rho_1^* - \pi_2 \rho_2^* \geq 0\}$, provided that $\pi_1 \rho_1^* - \pi_2 \rho_2^*$ is full rank.
\end{boxtheorem}

\begin{proof}
By Theorem~\ref{thm: Pemin minimax multiple}, we have the minimax equality:
\begin{align}
   \inf_{0\leq M \leq I} \sup_{\substack{\rho_1 \in \sC_1\\ \rho_2 \in \sC_2}}P_e(\{M,I-M\}, \{\rho_1, \rho_2\}) = \sup_{\substack{\rho_1 \in \sC_1\\ \rho_2 \in \sC_2}} \inf_{0\leq M \leq I}  P_e(\{M,I-M\}, \{\rho_1, \rho_2\}).
\end{align}
Lemma~\ref{lem: Pemin trace norm} states that for any $\rho_1, \rho_2$,
\begin{align}
    \inf_{0\leq M \leq I}  P_e(\{M,I-M\}, \{\rho_1, \rho_2\}) =  \frac{1}{2} \left( 1 - \frac{1}{2} \|\pi_1 \rho_1 - \pi_2 \rho_2\|_1 \right).
\end{align}
Since $(\rho_1^*, \rho_2^*)$ is the minimizer of the convex optimization $\min_{\rho_1 \in \sC_1,\, \rho_2 \in \sC_2} \|\pi_1 \rho_1 - \pi_2 \rho_2\|_1$, it also maximizes $\max_{{\rho_1 \in \sC_1, \rho_2 \in \sC_2}} [\inf_{0\leq M \leq I}  P_e(\{M,I-M\}, \{\rho_1, \rho_2\})]$.
Fixing $\rho_1^*, \rho_2^*$, we have
\begin{align}
    \inf_{0\leq M \leq I}  P_e(\{M,I-M\}, \{\rho_1^*, \rho_2^*\}) =  \pi_1 - \sup_{0\leq M \leq I} \tr M [\pi_1 \rho_1^* - \pi_2 \rho_2^*].
\end{align}
Since $\pi_1 \rho_1^* - \pi_2 \rho_2^*$ is assumed to be full rank, Lemma~\ref{lem: optimal test uniqueness} ensures that $\{\pi_1 \rho_1^* - \pi_2 \rho_2^* \geq 0\}$ is the unique maximizer of $\sup_{0\leq M \leq I} \tr[M(\pi_1 \rho_1^* - \pi_2 \rho_2^*)]$. Therefore, it is also the unique minimizer of $\inf_{0\leq M \leq I}  P_e(\{I-M,M\}, \{\rho_1^*, \rho_2^*\})$.
Finally, by Lemma~\ref{lem: saddle point minimax construction uniqueness}, $\{\pi_1 \rho_1^* - \pi_2 \rho_2^* \geq 0\}$ is the minimizer of $\inf_{0\leq M \leq I} [\max_{{\rho_1 \in \sC_1, \rho_2 \in \sC_2}}P_e(\{M,I-M\}, \{\rho_1, \rho_2\})]$. Therefore, it is the optimal test for distinguishing the sets $\sC_1$ and $\sC_2$.
\end{proof}

\section{Quantum Chernoff bound for two sets of quantum states}
\label{sec: Quantum Chernoff bound for two sets of quantum states}

In this section, we establish the quantum Chernoff bound for the discrimination of two sets of quantum states. We provide both lower and upper bounds on the asymptotic error exponent and show that, under suitable structural assumptions, the bounds coincide and the optimal exponent is given by the regularized Chernoff divergence between the sets.

\begin{boxassumption}[Assumptions for sets of quantum states.]
We denote the following assumptions for a sequence of sets of quantum states $\sC = \{\sC_n\}_{n\in \NN}$ where each $\sC_n \subseteq \density(\cH^{\ox n})$.
\begin{itemize}
    \item[(C1)] Convexity: For any $n \in \NN$, the sets $\sC_n$ are convex.
    \item[(C2)] Stability under tensor product: For any $m,n \in \NN$, it holds that $\sC_{m} \ox \sC_{n} \subseteq \sC_{m+n}$.
    \item[(C3)] Finiteness: $D_{\Petz, \alpha}(\sC_{1}\|\sC_{1}') < \infty$ for $\alpha \in (0,1)$ if considering two sequences $\sC,\sC'$.~\footnotemark
\end{itemize}
\end{boxassumption}

\footnotetext{This is a mild technical assumption, requiring that there exist $\rho \in \sC_1$ and $\rho'\in \sC_1'$ which are not orthogonal.}

\begin{boxtheorem}[Quantum Chernoff bound for binary hypotheses.]\label{thm: chernoff bound for two sets}
Let $\cH$ be a finite-dimensional Hilbert space, and let $\sC_i = \{\sC_{i,n}\}_{n\in \NN}$\, for $i \in \{1, 2\}$ be two sequences of sets of quantum states, where each $\sC_{i,n} \subseteq \density(\cH^{\otimes n})$. Let $\{\pi_1,\pi_2\}$ be the prior probability.
\begin{itemize}
    \item (Lower bound:) If the sequences $\sC_1$ and $\sC_2$ satisfy assumption (C1):
    \begin{align}
    \liminf_{n\to \infty} - \frac{1}{n} \log P_{e,\min}(\{\sC_{1,n}, \sC_{2,n}\}) \geq \underline{C}^\reg(\sC_1\|\sC_2).
    \end{align}
    \item (Upper bound:) If the sequences $\sC_1$ and $\sC_2$ satisfy assumption (C2) and (C3):
    \begin{align}
    \limsup_{n\to \infty} - \frac{1}{n} \log P_{e,\min}(\{\sC_{1,n}, \sC_{2,n}\}) \leq C^\reg(\sC_1\|\sC_2).
    \end{align}
\end{itemize}  
Consequently, if the sequences $\sC_1$ and $\sC_2$ satisfy assumptions (C1), (C2) and (C3), then the following limit exists and is given by
\begin{align}
    \lim_{n\to \infty} - \frac{1}{n} \log P_{e,\min}(\{\sC_{1,n}, \sC_{2,n}\}) = C^\reg(\sC_1\|\sC_2).
\end{align}
\end{boxtheorem}
\begin{remark}
    If we choose $\sC_{i,n}=\{\rho_i^{\ox n}\}$ as the singleton i.i.d. quantum states, then we can recover the  quantum Chernoff bound for two quantum states in~\cite{audenaert2007discriminating,Nussbaum2009}. Note that for general sets $\sC_{i,n}$, i.e., not necessarily convex, we have 
\begin{align}
 P_{e,\min}(\{\sC_{1,n},\sC_{2,n}\}) = P_{e,\min}(\{\conv(\sC_{1,n}),  \conv(\sC_{2,n})\}),
\end{align}
where $\conv(\sC)$ denotes the convex hull of set $\sC$. Therefore, we have 
\begin{align}
    \lim_{n\to \infty} - \frac{1}{n} \log P_{e,\min}(\{\sC_{1,n}, \sC_{2,n}\}) = C^\reg(\conv(\sC_1)\|\conv(\sC_2)),
\end{align}
where $\conv(\sC_i) = \{\conv(\sC_{i,n})\}_{n\in \NN}$ for $i\in\{1,2\}$.
This strengthens~\cite[Eq. (II.62)]{Mosonyi_2022} by showing that the result holds as an equality. 
\end{remark}

\subsection{Proof of lower bound}

For any $\alpha \in [0,1]$, we have that
\begin{align}
    P_{e, \min}(\{\rho_{1,n},\rho_{2,n}\}) & = \frac{1}{2} [\tr(\pi_2 \rho_{2,n} + \pi_1 \rho_{1,n}) - \|\pi_1 \rho_{1,n} - \pi_2 \rho_{2,n}\|_1] \\
    & \leq \pi_1^\alpha \pi_2^{1-\alpha} Q_{\alpha}(\{\rho_{1,n}, \rho_{2,n}\}),
\end{align}
where the equality follows from Lemma~\ref{lem: Pemin trace norm} and the inequality follows from Lemma~\ref{lem: audenaert}.
As this holds for any $\rho_{1,n} \in \sC_{1,n}$ and $\rho_{2,n} \in \sC_{2,n}$, we have from Theorem~\ref{thm: Pemin minimax multiple} and assumption (C1) that
\begin{align}
P_{e,\min}(\{\sC_{1,n},\sC_{2,n}\}) \leq (\pi_1^\alpha \pi_2^{1-\alpha})  Q_{\alpha}(\sC_{1,n}\| \sC_{2,n}) \leq Q_{\alpha}(\sC_{1,n}\| \sC_{2,n}),
\end{align}
where the second inequality follows from the fact that $0 \leq \pi_1^\alpha \pi_2^{1-\alpha} \leq \alpha \pi_1 + (1-\alpha) \pi_2 \leq 1$ for any $\alpha \in [0,1]$.   This gives
\begin{align}
    - \log P_{e,\min}(\{\sC_{1,n},\sC_{2,n}\}) \geq - \log  Q_{\alpha}(\sC_{1,n}\| \sC_{2,n}).
\end{align}
As this holds for any $\alpha \in [0,1]$, we have 
\begin{align}
    - \log P_{e,\min}(\{\sC_{1,n},\sC_{2,n}\}) & \geq \max_{\alpha \in [0,1]}- \log  Q_{\alpha}(\sC_{1,n}\| \sC_{2,n}) = C(\sC_{1,n}\| \sC_{2,n}),
\end{align}
where the equality follows Theorem~\ref{thm: C and Q between two sets}.
This implies 
\begin{align}
   \liminf_{n\to \infty} - \frac{1}{n}\log P_{e,\min}(\{\sC_{1,n},\sC_{2,n}\}) \geq \liminf_{n\to \infty} \frac{1}{n} C(\sC_{1,n}\| \sC_{2,n}) = \underline{C}^\reg(\sC_1\|\sC_2).
\end{align}

\subsection{Proof of upper bound}

We can easily prove the upper bound for limit inferior as follows. For any fixed $m \in \NN$ and any $\rho_{1,m} \in \sC_{1,m}, \rho_{2,m} \in \sC_{2,m}$, we have the following chain of inequalities:
\begin{align}
    \liminf_{n\to \infty} -\frac{1}{n}& \log P_{e,\min}(\{\sC_{1,n},\sC_{2,n}\})\notag\\
     & \leq \liminf_{n\to \infty} -\frac{1}{nm} \log P_{e,\min}(\{\sC_{1,nm},\sC_{2,nm}\})\\
    & \leq \liminf_{n\to \infty} -\frac{1}{nm} \log \sup_{\substack{\rho_{1,mn} \in \sC_{1,mn}\\ \rho_{2,mn} \in \sC_{2,mn}}}P_{e,\min}(\{\rho_{1,nm},\rho_{2,nm}\})\\
    & \leq \liminf_{n\to \infty} -\frac{1}{nm} \log P_{e,\min}(\{(\rho_{1,m})^{\ox n},(\rho_{2,m})^{\ox n}\})\\
    & = \frac{1}{m} C(\rho_{1,m}\|\rho_{2,m}),
\end{align}
where the first inequality follows as the lower limit of a subsequence is no smaller than the lower limit of the sequence, the second inequality holds trivially because for any $\rho_{1,n} \in \sC_{1,n}$ and $\rho_{2,n} \in \sC_{2,n}$, we have $P_{e,\min}(\{\sC_{1,n}, \sC_{2,n}\}) \geq P_{e,\min}(\{\rho_{1,n}, \rho_{2,n}\})$, the third inequality follows by taking a particular feasible solution and the assumption (C3), the equality follows from the quantum Chernoff bound between two quantum states (see Eq.~\eqref{eq: chernoff theorem}). 

As this holds for any $\rho_{1,m} \in \sC_{1,m}$ and $\rho_{2,m} \in \sC_{2,m}$, we have 
\begin{align}
     \liminf_{n\to \infty} -\frac{1}{n} \log P_{e,\min}(\{\sC_{1,n},\sC_{2,n}\}) \leq \frac{1}{m} C(\sC_{1,m}\|\sC_{2,m}).
\end{align}
As this holds for any $m \in \NN$, we have 
\begin{align}
    \liminf_{n\to \infty} -\frac{1}{n} \log P_{e,\min}(\{\sC_{1,n},\sC_{2,n}\}) \leq \liminf_{m \to \infty} \frac{1}{m} C(\sC_{1,m}\|\sC_{2,m}) = C^{\infty}(\sC_1\|\sC_2),
\end{align}
where the equality follows from the assumption (C2) and Remark~\ref{rem: limit existence for Chernoff}.

\bigskip
The above argument only gives the upper bound for limit inferior by choosing a suitable subsequence of i.i.d.~states. However, we show that the upper bound can be strengthened further to limit superior by carefully designing a sequence of states whose limit superior is also upper bounded by the regularized Chernoff divergence. However, as this sequence is not i.i.d. states anymore, its analysis is more challenging and requires the Nussbaum-Szko\l{}a distributions~\cite{Nussbaum2009} and the G\"{a}rtner-Ellis theorem~\cite{Dembo2010}.

Let the spectral decompositions of $\rho$ and $\sigma$ be given by 
\begin{align}
    \rho = \sum_{i=1}^d \lambda_i \ketbra{u_i}{u_i}\quad \text{and} \quad \sigma = \sum_{j=1}^d \mu_j \ketbra{v_j}{v_j},
\end{align}
where $\ket{u_i}$ and $\ket{v_j}$ are two orthonormal bases and $\lambda_i$ and $\mu_j$ are the corresponding eigenvalues, respectively. Then the Nussbaum-Szko\l{}a distributions of 
$\rho,\sigma$ are defined by
\begin{align}
(P_{\rho,\sigma})(i,j) = \lambda_i|\<u_i|v_j\>|^2 \quad \text{and} \quad (Q_{\rho,\sigma})(i,j) = \mu_j |\<u_i|v_j\>|^2,
\end{align}
where $i,j \in \{1,\cdots, d\}$. 

In the remaining discussion of this section, let $\log$ be a logarithm with natural base $e$ for simplicity.
Given a sequence of random variables $\{X_n\}_{n\in \NN}$, the asymptotic cumulant generating function is defined as
\begin{align}
    \Lambda_X(t):= \lim_{n\to \infty} \frac{1}{n} \log \mathbb{E}\left[\exp({nt X_n})\right],
\end{align}
provided that the limit exists. For our purpose, it is sufficient to use the following variant of the G\"{a}rtner-Ellis theorem due to~\cite[Theorem 3.6]{chen2000GeneralizationGartnerEllisTheorem} (see also~\cite[Proposition 17]{hayashi2016correlation}).

\begin{lemma}\label{lem: Gartner-Ellis}
    Assume that the asymptotic cumulant generating function $t\mapsto \Lambda_X(t)$ exists and is differentiable in some interval $(a,b)$. Then, for any $x \in (\lim_{t\to a^+} \Lambda_X'(t), \lim_{t\to b^-} \Lambda_X'(t))$,  
    \begin{align}
        \limsup_{n\to \infty} -\frac{1}{n} \log \Pr\{X_n \geq x\} \leq \sup_{t\in (a,b)} \{tx - \Lambda_X(t)\}.
    \end{align}
\end{lemma}

\begin{boxlemma}\label{LH1}
Let $m \in \NN$ be any integer and $\{\pi_1,\pi_2\}$ be a prior probability. Let 
$\rho_{1,1} \in \sC_{1,1}$, $\rho_{2,1} \in \sC_{2,1}$ 
and
$\rho_{1,m} \in \sC_{1,m}$, $\rho_{2,m} \in \sC_{2,m}$ be quantum states such that $D_{\Petz,\alpha}(\rho_{1,1}\|\rho_{2,1})<\infty $ and $D_{\Petz,\alpha}(\rho_{1,m}\|\rho_{2,m})<\infty $ for any $\alpha \in (0,1)$. We set $k:=\lfloor n/m\rfloor$ and construct quantum states
\begin{align}
    \rho_1^{(n)} :=
\rho_{1,1}^{\otimes n-km}\otimes \rho_{1,m}^{\otimes k}\quad \text{and} \quad
\rho_2^{(n)}:=
\rho_{2,1}^{\otimes n-km}\otimes \rho_{2,m}^{\otimes k}.
\end{align} 
Then
\begin{align}
\limsup_{n\to \infty} -\frac{1}{n}\log P_{e,\min}\left(\{\rho_1^{(n)},\rho_2^{(n)}\}\right) \leq \frac{1}{m}C(\rho_{1,m}\|\rho_{2,m}).
\end{align}
\end{boxlemma}

\begin{proof}
Let ${P^{(n)}}$ and ${Q^{(n)}}$ be the Nussbaum-Szko\l{}a distributions of $\rho_1^{(n)}$ and $\rho_2^{(n)}$. Let 
\begin{align} 
    S_n = \left\{{\pi_1 P^{(n)}} - \pi_2 {Q^{(n)}} > 0\right\},
\end{align} 
be the maximum likelihood test. Consider the random variable
\begin{align}
    X_n(x) := \frac{1}{n} \left(\log \left[\pi_2 {Q^{(n)}}(x)\right] - \log \left[\pi_1 {P^{(n)}}(x)\right]\right),
\end{align}
where $x$ is drawn from the distribution ${P^{(n)}}$. Let $\phi(s):= \frac{1}{m}\log (\rho_{1,m})^{1-s} (\rho_{2,m})^{s}$. 
Then we have the asymptotic cumulant generating function of the random variable $X_n$ as,
\begin{align}
\lim_{n\to \infty} &\frac{1}{n}\log \sum_{x} {P^{(n)}}(x) \exp \left(s n X_n(x) \right)\notag\\
&  = \lim_{n\to \infty}\frac{1}{n}\log \tr {(Q^{(n)})}^s {(P^{(n)})}^{1-s}\\
& = \lim_{n\to \infty} \frac{1}{n} \log \tr (\rho_2^{(n)})^s (\rho_1^{(n)})^{1-s}\\
& =
\lim_{n\to \infty}\frac{1}{n}
\left((n-km) \log \tr (\rho_{1,1})^{1-s} (\rho_{2,1})^{s}
+k \log \tr (\rho_{1,m})^{1-s} (\rho_{2,m})^{s}\right)\\
& =\phi(s),
\end{align}
where the second equality is a simple fact of the Nussbaum-Szko\l{}a distributions (e.g.~\cite[Proposition 1]{audenaert2008asymptotic}). Note that $\phi(s)$ is differentiable and $\phi'(0) = - \frac{1}{m}D(\rho_{1,m}\|\rho_{2,m})$ and $\phi'(1) = \frac{1}{m}D(\rho_{2,m}\|\rho_{1,m})$.  Applying the G\"{a}rtner-Ellis theorem in Lemma~\ref{lem: Gartner-Ellis} for the random variable $X_n$, interval $(0,1)$ and $x=0$, we have
\begin{align}
    \limsup_{n\to \infty} - \frac{1}{n} \log \Pr\{X_n \geq 0\} \leq \sup_{s\in (0,1)} -\phi(s). \label{eq: NS relation tmp1}
\end{align}
Similarly, consider the random variable 
\begin{align} 
    Y_n(x):=\frac{1}{n} \left(\log \left[\pi_1 {P^{(n)}}(x)\right] - \log \left[\pi_2 {Q^{(n)}}(x)\right]\right),
\end{align} 
where $x$ is drawn from the distribution ${Q^{(n)}}$. Then we have the asymptotic cumulant generating function of the random variable $Y_n$ as,
\begin{align}
\lim_{n\to \infty} &\frac{1}{n}\log \sum_{x} {Q^{(n)}}(x) \exp \left(t n Y_n(x)\right)\notag\\
&  = \lim_{n\to \infty}\frac{1}{n}\log \tr {(Q^{(n)})}^{1-t} {(P^{(n)})}^{t}\\
& = \lim_{n\to \infty} \frac{1}{n} \log \tr (\rho_2^{(n)})^{1-t} (\rho_1^{(n)})^{t}\\
& =
\lim_{n\to \infty}\frac{1}{n}
\left((n-km) \log \tr (\rho_{1,1})^{t} (\rho_{2,1})^{1-t}
+k \log \tr (\rho_{1,m})^{t} (\rho_{2,m})^{1-t}\right)\\
& =\phi(1-t).
\end{align}
Applying again the G\"{a}rtner-Ellis theorem in Lemma~\ref{lem: Gartner-Ellis} for the random variable $Y_n$, interval $(0,1)$ and $x=0$, we have
\begin{align}
\limsup_{n\to \infty} -\frac{1}{n} \log \Pr\{Y_n \geq 0\} & \leq \sup_{t \in (0,1)} -\phi(1-t).\label{eq: NS relation tmp2}
\end{align}
By direct calculation, we have the relations
\begin{align}
    \limsup_{n\to \infty} - \frac{1}{n} \log \Pr\{X_n \geq 0\} & = \limsup_{n\to \infty}  -\frac{1}{n}\log \tr P^{(n)} (I-S_n),\label{eq: NS relation tmp3}\\
\limsup_{n\to \infty} - \frac{1}{n} \log \Pr\{Y_n \geq 0\} & = \limsup_{n\to \infty}  -\frac{1}{n}\log \tr Q^{(n)} S_n. \label{eq: NS relation tmp4}
\end{align}
Moreover, the Nussbaum-Szko\l{}a theorem (e.g.~\cite[Lemma 3.4]{hayashi2017quantum}) implies that for any test $T_n$,
\begin{align}
\pi_1 \tr  \rho_1^{(n)} (I-T_n)  & + \pi_2 \tr \rho_2^{(n)} T_n \geq \frac{1}{2}\left(
\pi_1 \tr  {P^{(n)}} (I- S_n) + \pi_2 \tr {Q^{(n)}} S_n\right).
\end{align}
So we have 
\begin{align}\label{eq: NS relation tmp5}
   P_{e,\min}\left(\{\rho_1^{(n)},\rho_2^{(n)}\}\right) \geq \frac{1}{2}\left(
\pi_1 \tr  {P^{(n)}} (I- S_n) + \pi_2 \tr {Q^{(n)}} S_n\right).
\end{align}
Combining Eqs.~\eqref{eq: NS relation tmp1},~\eqref{eq: NS relation tmp2},~\eqref{eq: NS relation tmp3},~\eqref{eq: NS relation tmp4} and~\eqref{eq: NS relation tmp5},
we have
\begin{align}
\limsup_{n\to \infty} & -\frac{1}{n}\log P_{e,\min}\left(\{\rho_1^{(n)},\rho_2^{(n)}\}\right)\notag \\
& \leq \limsup_{n\to \infty} - \frac{1}{n}\log \frac{1}{2} \left[\pi_1\tr P^{(n)} (I-S_n) +\pi_2\tr {Q^{(n)}} S_n\right] \\
& \leq \min \left\{ \sup_{s\in (0,1)} -\phi(s),
\sup_{t \in (0,1)}-\phi(1-t)\right\}\\
& = \sup_{s\in (0,1)} -\phi(s)\\
& = \frac{1}{m}C(\rho_{1,m}\|\rho_{2,m}),\label{eq: lemma tmp1}
\end{align}
where the last equality follows by replacing $t$ with $1-s$. 
\end{proof}

With the above Lemma~\ref{LH1}, we are ready to show the upper bound for limit superior. That is, we aim to 
show that
    \begin{align}
        \limsup_{n\to \infty} -\frac{1}{n} \log P_{e,\min}(\{\sC_{1,n},\sC_{2,n}\})  \le C^\infty(\sC_1\|\sC_2).
    \end{align}
For this, we plan to show that for any fixed $m \in \NN$,
    \begin{align}\label{eq: upper bound proof tmp2}
        \limsup_{n\to \infty} -\frac{1}{n} \log P_{e,\min}(\{\sC_{1,n},\sC_{2,n}\})  \le \frac{1}{m}C(\sC_{1,m}\|\sC_{2,m}).
    \end{align}

If $C(\sC_{1,m}\|\sC_{2,m}) = \infty$, the upper bound holds trivially. It remains to show Eq.~\eqref{eq: upper bound proof tmp2} when $C(\sC_{1,m}\|\sC_{2,m}) < \infty$. So for any $\delta > 0$, there exist $\rho_{1,m} \in \sC_{1,m}$ and $\rho_{2,m} \in \sC_{2,m}$ such that
\begin{align}\label{eq: choice of rhom and sigmam}
    C(\rho_{1,m}\|\rho_{2,m}) \leq C(\sC_{1,m}\|\sC_{2,m}) + \delta < \infty.
\end{align}
This implies that $D_{\Petz, \alpha}(\rho_{1,m}\|\rho_{2,m}) < \infty$ for any $\alpha \in (0,1)$. Otherwise, we will have a contradiction to the finiteness of $C(\rho_{1,m}\|\rho_{2,m})$ by definition in Eq.~\eqref{eq: definition of Chernoff divergence}.

Using this choice of $\rho_{1,m}$ and $\rho_{2,m}$ and taking any $\rho_{1,1},\rho_{2,1}$ such that $D_{\Petz,\alpha}(\rho_{1,1}\|\rho_{2,1}) < \infty$ (the existence of these states is ensured by assumption (C3)), we construct the sequence of states $\rho_1^{(n)}$ and $\rho_2^{(n)}$ as in Lemma~\ref{LH1} and we have 
\begin{align}
    \limsup_{n\to \infty} & -\frac{1}{n}\log P_{e,\min}\left(\{\rho_1^{(n)},\rho_2^{(n)}\}\right) \leq \frac{1}{m} C(\rho_{1,m}\|\rho_{2,m}),
\end{align}
This implies that 
\begin{align}
\limsup_{n\to \infty} -\frac{1}{n}\log P_{e,\min}(\{\sC_{1,n},\sC_{2,n}\}) & \leq \limsup_{n\to \infty} -\frac{1}{n}\log P_{e,\min}(\{\rho_{1}^{(n)}, \rho_{2}^{(n)}\})\\
& \leq \frac{1}{m} C(\rho_{1,m}\|\rho_{2,m})\\
 & \leq \frac{1}{m}C(\sC_{1,m}\|\sC_{2,m}) + \frac{\delta}{m},
\end{align}
where the first line follows because $P_{e,\min}(\{\sC_{1,n},\sC_{2,n}\}) \geq P_{e,\min}(\{\rho_{1,n},\rho_{2,n}\})$ for any $\rho_{1,n} \in \sC_{1,n}$ and $\rho_{2,n} \in \sC_{2,n}$, the last line follows by the choice of $\rho_{1,m}$ and $\rho_{2,m}$ in Eq.~\eqref{eq: choice of rhom and sigmam}.

As this holds for any $\delta > 0$, we have 
\begin{align}
\limsup_{n\to \infty} -\frac{1}{n}\log P_{e,\min}(\{\sC_{1,n},\sC_{2,n}\})\leq \frac{1}{m}C(\sC_{1,m}\|\sC_{2,m}).  
\end{align}
Finally, taking the limit $m\to \infty$, we have 
    \begin{align}
        \limsup_{n\to \infty} -\frac{1}{n} \log P_{e,\min}(\{\sC_{1,n},\sC_{2,n}\})  \le C^\infty(\sC_1\|\sC_2),
    \end{align}
 where the existence of the limit on the right hand side follows from the stability assumption (C3) and Eq.~\eqref{eq: multiplicativity of Chernoff divergence}.
This finishes the proof of the upper bound.

\section{Quantum Chernoff bound for multiple sets of quantum states}
\label{sec: Quantum Chernoff bound for multiple sets of quantum states}

In this section, we generalize the quantum Chernoff bound from the binary case to the discrimination of multiple sets of quantum states. Analogous to the result for multiple quantum states in Eq.~\eqref{eq: chernoff theorem for multiple states}, we demonstrate that the optimal error exponent is characterized by the smallest regularized Chernoff divergence among all pairs of sets. The proof of the lower bound (achievability) is more involved than in the binary case, requiring a discretization technique to control the spectra of the states within the sets. The proof of the upper bound (converse) reduces the problem to the binary case by converting any test for multiple hypotheses into a test between two hypotheses.

\begin{boxtheorem}[Quantum Chernoff bound for multiple hypotheses.]\label{thm: chernoff bound for multiple sets}
Let $\cH$ be a finite-dimensional Hilbert space, and let $\sC_i = \{\sC_{i,n}\}_{n\in \NN}$\, for $i \in \{1, \ldots, r\}$ be $r$ sequences of sets of quantum states, where each  $\sC_{i,n} \subseteq \density(\cH^{\otimes n})$. Let $\{\pi_i\}_{i=1}^r$ be the prior probability. 
\begin{itemize}
    \item (Lower bound:) If the sequences $\{\sC_i\}_{i=1}^r$ satisfy assumption (C1): 
    \begin{align}
    \liminf_{n\to \infty} - \frac{1}{n} \log P_{e,\min}(\{\sC_{i,n}\}_{i=1}^r) \geq \min_{i\neq j} \underline{C}^\infty(\sC_{i}\|\sC_{j}).
    \end{align}
    \item (Upper bound:) If the sequences $\{\sC_i\}_{i=1}^r$ satisfy assumption (C2) and (C3):
    \begin{align}
    \limsup_{n\to \infty} - \frac{1}{n} \log P_{e,\min}(\{\sC_{i,n}\}_{i=1}^r) \leq \min_{i\neq j} C^\infty(\sC_{i}\|\sC_{j}).
    \end{align}
\end{itemize}
Consequently, if the sequences $\{\sC_i\}_{i=1}^r$ satisfy assumptions (C1), (C2) and (C3), then the following limit exists and is given by
\begin{align}
    \lim_{n\to \infty} - \frac{1}{n} \log P_{e,\min}(\{\sC_{i,n}\}_{i=1}^r) = \min_{i\neq j} C^\infty(\sC_{i}\|\sC_{j}).
\end{align}
\end{boxtheorem}

\begin{remark}
    If we choose $\sC_{i,n}=\{\rho_i^{\ox n}\}$ as the singleton i.i.d. quantum states, then we can recover the  quantum Chernoff bound for multiple quantum states in~\cite{li2016discriminating}. Note that for general sets $\sC_{i,n}$, i.e., not necessarily convex, we have 
\begin{align}
 P_{e,\min}(\{\sC_i\}_{i=1}^r) = P_{e,\min}(\{\conv(\sC_i)\}_{i=1}^r),
\end{align}
where $\conv(\sC)$ denotes the convex hull of set $\sC$. Therefore, we have 
\begin{align}
    \lim_{n\to \infty} - \frac{1}{n} \log P_{e,\min}(\{\sC_{i,n}\}_{i=1}^r) = \min_{i\neq j} C^\infty(\conv(\sC_{i})\|\conv(\sC_{j})),
\end{align}
where $\conv(\sC_{i}) = \{\conv(\sC_{i,n})\}_{n\in \NN}$.
\end{remark}

\subsection{Proof of lower bound}

As the proof is a bit long, we divide it into a few steps.

\paragraph{Step 1)} For any $\delta > 0$, let $\rho^*_{i,n}\in \sC_{i,n}$ be quantum states such that
\begin{align}\label{eq: choice of rho star}
    P_{e,\min}(\{\sC_{i,n}\}_{i=1}^r) = \sup_{\forall i,\, \rho_i\in \sC_i}  P_{e,\min}(\{\rho_i\}_{i=1}^r) \leq P_{e,\min}(\{\rho^*_{i,n}\}_{i=1}^r) + \delta,
\end{align}
where the equality follows from Theorem~\ref{thm: Pemin minimax multiple} and assumption (C1). 
Our goal is to upper bound $P_{e,\min}(\{\rho^*_{i,n}\}_{i=1}^r)$ using the result from~\cite[Theorem 2]{li2016discriminating} for a collection of states. However, the upper bound depends on the spectra of the states to discriminate. To address this, we modify each state $\rho_{i,n}^*$ to obtain a new state with a finite size spectrum, while ensuring that the modification does not significantly affect the relevant quantities. This is achieved via a spectrum discretization technique, similar to those used in~\cite{hayashi2014large,hayashi2024generalized}.

For any quantum states $\rho_i$ with spectral decomposition
$\rho_i=\sum_{k} r_{i,k}E_{i,k}$, we define
\begin{align}
Q(\rho_{i}\|\rho_{j})
:=\sum_{k,l}    \min\{ r_{i,k}, r_{j,l}\}\tr E_{i,k} E_{j,l},
\end{align}
and the generalization to sets of quantum states as
\begin{align}
    Q(\sC_{i}\|\sC_{j}):= \sup_{\substack{\rho_{i} \in \sC_{i}\\ \rho_{j} \in \sC_{j}}} Q(\rho_{i}\|\rho_{j}).
\end{align}
Observe that for any real numbers $a, b > 0$, we have $\min\{a, b\} \leq \min\{a^\alpha b^{1-\alpha} : \alpha \in [0,1]\}$. Applying this to the definition of $Q(\rho_{i}\|\rho_{j})$, it follows that
\begin{align}
Q(\rho_{i}\|\rho_{j}) \leq \min_{\alpha \in [0,1]} Q_\alpha(\rho_{i}\|\rho_{j}) =  2^{- C(\rho_{i}\|\rho_{j})}.
\end{align}
Consequently,
\begin{align}\label{eq: Q C relation}
    Q(\sC_{i}\|\sC_{j}) \leq 2^{- C(\sC_{i}\|\sC_{j})}.
\end{align}

We define $d := \dim {\cal H}$ and choose $\lambda_n>0$ as
    \begin{align}\label{eq: choice of lambda}
\lambda_n:= \frac{1}{n}\min_{i\neq j} C(\sC_{i,n}\|\sC_{j,n})+2 \log d.
    \end{align}
Define the function for discretization as
\begin{align}
    g_n(x):= 
\begin{cases}    
    & \left\lfloor \frac{x}{\lambda_n}\right \rfloor 
    \lambda_n \quad \quad \ \hbox{if } 0\leq x \le n \lambda_n,\\
& \quad n\lambda_n  \quad \quad  \quad \hbox{if }\quad \ \ \ x > n \lambda_n.
\end{cases}    
\end{align}
Then the range of $g_n$ is a discrete set $\{0, \lambda_n, 2\lambda_n, \ldots, n\lambda_n\}$.
Moreover, for $x \le n \lambda_n$, we have 
\begin{align}
    x-\lambda_n \leq g_n(x) \leq x.\label{VH1}
\end{align}
For $x > n \lambda_n$, we have only the following relation
\begin{align}
n \lambda_n =   g_n(x) \leq x.\label{VH2}
\end{align}

Let 
\begin{align} \label{eq: rho in spectrum decomposition}
    \rho_{i,n}^*=\sum_{k} 2^{- r_{i,n,k}}E_{i,n,k},
\end{align} 
be the spectrum decomposition with $ E_{i,n,k}$ being rank-one projections.
Then, we define the modified quantum state as
\begin{align} \label{eq: definition of modified state}
    \rho_{i,n}' := \frac{1}{c_{i,n}}\sum_k 2^{ -g_n(r_{i,n,k})} E_{i,n,k},
\end{align}
with the normalization factor
\begin{align}
    c_{i,n}:=& \tr \left[\sum_k 2^{ -g_n(r_{i,n,k})} E_{i,n,k}\right].
\end{align}

For the normalization factor, we have
\begin{align} 
c_{i,n} 
= & \tr \left[\sum_{k \,:\, r_{i,n,k} \le n \lambda_n} 2^{ -g_n(r_{i,n,k})} E_{i,n,k}\right]
+ \tr \left[\sum_{k \,:\, r_{i,n,k} > n \lambda_n} 2^{ -g_n(r_{i,n,k})} E_{i,n,k}\right] \\
\le &
 \tr \left[\sum_{k \,:\, r_{i,n,k} \le n \lambda_n} 2^{ -r_{i,n,k} +\lambda_n} E_{i,n,k}\right]
+ \tr \left[\sum_{k \,:\, r_{i,n,k} > n \lambda_n} 2^{-n \lambda_n}E_{i,n,k} \right]\\
\le &
 2^{ \lambda_n} \tr \rho_{i,n}^*
+ 2^{ -n \lambda_n} d^n \\
= & 2^{ \lambda_n} + 2^{ -n \lambda_n} d^n  \\
\le & 2^{ \lambda_n} + d^{-n},
\end{align}
where the first line splits the summation into two parts, the second line applies the bounds from Eq.~\eqref{VH1} and Eq.~\eqref{VH2}, the third line uses the fact that the first sum is a partial sum over the spectrum decomposition of $\rho_{i,n}^*$ while the second sum is bounded by the dimension $d^n$, and the last line follows from the choice of $\lambda_n$, which ensures $\lambda_n \geq 2\log d$.

On the other hand, we also have 
\begin{align}
c_{i,n} \rho_{i,n}'=& \sum_k 2^{ -g_n(r_{i,n,k})} E_{i,n,k}
\ge  \sum_k 2^{ -r_{i,n,k}} E_{i,n,k} = \rho_{i,n}^*,
\end{align}
which then implies that 
\begin{align}\label{eq: cin geq 1}
    c_{i,n} \geq 1.
\end{align}

\paragraph{Step 2)}
Now we apply the result~\cite[Theorem 2]{li2016discriminating} to a collection of unnormalized states $\{c_{i,n} \rho_{i,n}'\}_{i=1}^r$ to upper bound $P_{e,\min}(\{\sC_{i,n}\}_{i=1}^r)$.
Since $\rho_{i,n}^* \leq c_{i,n} \rho_{i,n}'$,  we have
\begin{align}
    P_{e,\min}(\{\sC_{i,n}\}_{i=1}^r) & \leq P_{e,\min}(\{\rho_{i,n}^{*}\}_{i=1}^r) + \delta \leq  P_{e,\min}(\{c_{i,n} \rho'_{i,n}\}_{i=1}^r) + \delta.
\label{eq: multiple lower bound proof tmp1}
\end{align}

It is known from~\cite[Theorem 2]{li2016discriminating} that for any collection of unnormalized quantum states $\{\rho_i\}_{i=1}^r$ with prior probabilities $\{\pi_i\}_{i=1}^r$,
the minimum error probability admits the upper bound
\begin{align}\label{eq: chernoff state P upper bound}
    P_{e,\min}(\{\rho_i\}_{i=1}^r) \leq  f(\{\rho_i\}_{i=1}^r) \sum_{i<j} 
      Q(\rho_{i}\|\rho_{j}),
\end{align}
where the coefficient is given by
\begin{align} 
   f(\{\rho_i\}_{i=1}^r) &:= 10r^2 \left(\max_{i} |\spec(\rho_i)|\right)^2,
\end{align} 
and $|\spec(X)|$ denotes the number of distinct eigenvalues of $X$. 

Applying this result to the states $\{c_{i,n}\rho_{i,n}'\}_{i=1}^r$, we have the upper bound
\begin{align}\label{eq: multiple lower bound proof tmp2}
    P_{e,\min}(\{c_{i,n}\rho_{i,n}'\}_{i=1}^r) 
    \leq  10r^4(n+1)^2  \max_{i\neq j}  
    Q(c_{i,n}\rho'_{i,n}\|c_{j,n}\rho'_{j,n}) ,
\end{align}
where we relax the summation in Eq.~\eqref{eq: chernoff state P upper bound} by the maximum value times the number of terms $C_r^2 = \frac{r(r-1)}{2} \leq r^2$ and use the fact that the number of distinct eigenvalues of $c_{i,n}\rho_{i,n}'$ is at most $n+1$ due to the discretization of $g_n$.

\paragraph{Step 3)}

Next, we upper bound $Q(c_{i,n}\rho_{i,n}'\|c_{j,n}\rho_{j,n}')$ by $2^{-n \min_{i\neq j} C(\sC_{i,n}\|\sC_{j,n})}$ for any fixed $i\neq j$. This shows that our discretization procedure does not incur any loss in the asymptotic exponent.

For this, we have by definition that
\begin{align}\label{eq: Q rho prime sum tmp1}
  Q(c_{i,n}\rho_{i,n}'\|c_{j,n}\rho_{j,n}')  
   =&\sum_{k,l} 
\min  \left\{  2^{ -g_n(r_{i,n,k})} , 2^{ -g_n(r_{j,n,l})}\right\} \tr E_{i,n,k} E_{j,n,l}.
\end{align} 
The summation in Eq.~\eqref{eq: Q rho prime sum tmp1} can be split into two parts. Denote the sets of indices:
\begin{align}
    S_1:= & \Big\{(k,l): \max\{r_{i,n,k},r_{j,n,l}\} > n\lambda_n\Big\},\\
    S_2:= & \Big\{(k,l): \max\{r_{i,n,k},r_{j,n,l}\} \leq n\lambda_n\Big\}.
\end{align}

For any $(k,l) \in S_1$, we have 
\begin{align}
\min  \left\{  2^{ -g_n(r_{i,n,k})} , 2^{ -g_n(r_{j,n,l})}\right\} = & 2^{ -g_n\left(\max\{r_{i,n,k},r_{j,n,l}\}\right)} \leq 2^{-n\lambda_n}.
\label{BNR}
\end{align}
This implies that 
\begin{align}
\sum_{(k,l)\in S_1}  
\min  &\left\{  2^{ -g_n(r_{i,n,k})} , 2^{ -g_n(r_{j,n,l})}\right\} \tr E_{i,n,k} E_{j,n,l} \notag \\
& \leq  2^{-n\lambda_n}  \sum_{(k,l)\in S_1}
\tr E_{i,n,k} E_{j,n,l} \leq d^{2n} 2^{-n\lambda_n} = 2^{-n \min_{i\neq j} C(\sC_{i,n}\|\sC_{j,n})}, \label{eq: Q rho prime sum tmp2} \end{align}
where we use the fact that the total number of terms in the summation is at most $d^{2n}$ and each term $\tr E_{i,n,k} E_{j,n,l} \leq 1$ as $E_{i,n,k}$ and $E_{j,n,l}$ are rank-one projections, and the equality follows from the choice of $\lambda_n$ in Eq.~\eqref{eq: choice of lambda}.

For any $(k,l) \in S_2$, we have by Eq.~\eqref{VH1} that   
\begin{align}
\min & \left\{  2^{ -g_n(r_{i,n,k})} , 2^{ -g_n(r_{j,n,l})}\right\}\notag\\
&\hspace{2cm} \leq  \min  \left\{  2^{ -(r_{i,n,k}-\lambda_n)} , 2^{ -(r_{j,n,l}-\lambda_n)}\right\}
= 2^{ \lambda_n} \min\left\{2^{-r_{i,n,k}}, 2^{-r_{j,n,l}}\right\}.
\end{align}
This implies that 
\begin{align}
\sum_{(k,l) \in S_2} & 
\min  \left\{  2^{ -g_n(r_{i,n,k})} , 2^{ -g_n(r_{j,n,l})}\right\} \tr E_{i,n,k} E_{j,n,l} \notag \\
& \leq  2^{ \lambda_n}  \sum_{(k,l) \in S_2}
\min\left\{2^{-r_{i,n,k}}, 2^{-r_{j,n,l}}\right\} \tr E_{i,n,k} E_{j,n,l}\\
& \leq 2^{ \lambda_n}  Q(\rho_{i,n}^*\|\rho_{j,n}^*)\\
& \leq 2^{\lambda_n} Q(\sC_{i,n}\|\sC_{j,n})\\
& \leq 2^{\lambda_n} 2^{-n C(\sC_{i,n}\|\sC_{j,n})}\\
& \leq 2^{\lambda_n} 2^{-n \min_{i\neq j} C(\sC_{i,n}\|\sC_{j,n})},\label{eq: Q rho prime sum tmp3}
\end{align}
where we relax the summation by including all terms in the second inequality and use Eq.~\eqref{eq: Q C relation} in the fourth inequality.

Combining Eqs.~\eqref{eq: Q rho prime sum tmp1},~\eqref{eq: Q rho prime sum tmp2} and~\eqref{eq: Q rho prime sum tmp3}, we have
\begin{align}
    Q(c_{i,n}\rho_{i,n}'\|c_{j,n}\rho_{j,n}')  
    \leq (1+2^{\lambda_n}) 2^{-n \min_{i\neq j} C(\sC_{i,n}\|\sC_{j,n})}
   \label{eq: multiple lower bound proof tmp3}.
\end{align}

\paragraph{Step 4)}
Combining Eqs.~\eqref{eq: multiple lower bound proof tmp1},~\eqref{eq: multiple lower bound proof tmp2} and~\eqref{eq: multiple lower bound proof tmp3}, we have
\begin{align}
P_{e,\min}(\{\sC_{i,n}\}_{i=1}^r) {\le}& P_{e,\min}(\{c_{i,n} \rho'_{i,n}\}_{i=1}^r) + \delta \\
         {\le}& 10r^4(n+1)^2  \max_{i\neq j}  
    Q(c_{i,n}\rho'_{i,n}\|c_{j,n}\rho'_{j,n})+ \delta \\
         {\le} &10r^4(n+1)^2(1+2^{\lambda_n}) 2^{-n \min_{i\neq j} C(\sC_{i,n}\|\sC_{j,n})}
+ \delta.
\end{align}

Since $\delta $ is arbitrary, we have
\begin{align}
P_{e,\min}(\{\sC_{i,n}\}_{i=1}^r) \leq &
         10r^4(n+1)^2 (1+2^{\lambda_n}) 2^{-\min_{i\neq j}  C(\sC_{i,n}\|\sC_{j,n})}.
\label{BJI}
\end{align}

\paragraph{Step 5)}
Finally, we show that the prefactor $10r^4(n+1)^2 (1+2^{\lambda_n})$ does not affect the asymptotic error exponent in the large $n$ limit.

Since $\lambda_n > 0$, we have  
\begin{align} 
    -\log \left(1+2^{\lambda_n}\right) \geq -\log (2^{\lambda_n} + 2^{\lambda_n}) = -\log (2^{\lambda_n+1}) = -(\lambda_n +1).
\end{align} 
This gives 
\begin{align}
    -\log \left(1+2^{\lambda_n}\right) & \geq - \left(\frac{1}{n}\min_{i\neq j} C(\sC_{i,n}\|\sC_{j,n}) + 2\log d + 1\right).
\end{align}

The relation in Eq.~\eqref{BJI} is converted to 
\begin{align}
    - \log & P_{e,\min}(\{\sC_{i,n}\}_{i=1}^r)\notag \\
     & \geq -\log \left[
    10r^4(n+1)^2   (1+2^{\lambda_n})
    \right] +\min_{i\neq j}  C(\sC_{i,n}\|\sC_{j,n})\\
    & \geq -\log \left[
    10r^4(n+1)^2)\right] + \left(1-\frac{1}{n}\right) \min_{i\neq j} C(\sC_{i,n}\|\sC_{j,n}) - 2\log d - 1
   \end{align}
Dividing both sides by $n$ and taking the limit inferior $n\to \infty$, we have
\begin{align}
\liminf_{n\to \infty} - \frac{1}{n}\log P_{e,\min}(\{\sC_{i,n}\}_{i=1}^r)  
   \geq  & 
   \liminf_{n\to \infty}     \frac{1}{n}\min_{i\neq j}C(\sC_{i,n}\|\sC_{j,n}),\label{eq: liminf lower bound 1}
\end{align}
because the term $-\frac{1}{n}(\log \left[10r^4(n+1)^2 \right] + 2\log d + 1)$
vanishes in the limit $n\to \infty$.  

Let $n_k$ be a subsequence such that 
\begin{align}
    \liminf_{n\to \infty}   \frac{1}{n}\min_{i\neq j}C(\sC_{i,n}\|\sC_{j,n}) = \lim_{k\to \infty}     \frac{1}{n_k}\min_{i\neq j}C(\sC_{i,n_k}\|\sC_{j,n_k}). \label{eq: liminf lower bound 2}
\end{align}
Let $(i^*(n_k), j^*(n_k))$ be the pair of indices achieving the minimum in $ \min_{i\neq j}C(\sC_{i,n_k}\|\sC_{j,n_k})$ for each $n_k$. Since there are only finitely many pairs of indices in $\{1,\cdots,r\}\times \{1,\cdots, r\}$, there exists a pair of $(i', j')$ which is visited by $(i^*(n_k), j^*(n_k))$ finitely many times. Consider the subsequence $\{n_{k_l}: (i^*(n_{k_l}), j^*(n_{k_l})) = (i',j')\}$. Therefore, we have
\begin{align}
   \lim_{k\to \infty}     \frac{1}{n_k}\min_{i\neq j}C(\sC_{i,n_k}\|\sC_{j,n_k}) = & \lim_{k\to \infty} \frac{1}{n_k} C(\sC_{i^*(n_k), n_k}\|\sC_{j^*(n_k), n_k})\\
   =&  \limsup_{k\to \infty} \frac{1}{n_k} C(\sC_{i^*(n_k), n_k}\|\sC_{j^*(n_k), n_k})\\
   \geq & \limsup_{l\to \infty} \frac{1}{n_{k_l}} C(\sC_{i^*(n_{k_l}), n_{k_l}}\|\sC_{j^*(n_{k_l}), n_{k_l}})\\
   = & \limsup_{l\to \infty} \frac{1}{n_{k_l}} C(\sC_{i', n_{k_l}}\|\sC_{j', n_{k_l}})\\
\geq  & \liminf_{l\to \infty} \frac{1}{n_{k_l}} C(\sC_{i', n_{k_l}}\|\sC_{j', n_{k_l}})\\
\geq & \liminf_{n\to \infty} \frac{1}{n} C(\sC_{i', n}\|\sC_{j', n}).\\
   \geq & \min_{i\neq j}
   \liminf_{n\to \infty}     \frac{1}{n}C(\sC_{i,n}\|\sC_{j,n})\\
   = & \min_{i\neq j}
  \underline{C}^{\infty}(\sC_{i}\|\sC_{j}),\label{eq: liminf lower bound 3}
\end{align}
where the first equality follows from the definition of $(i^*(n_k), j^*(n_k))$, the second line follows as the limit exists, the third line follows by considering a subsequence, the fourth line follows from the definition of the subsequence $n_{k_l}$, the fifth line holds as the limit superior is no smaller than the limit inferior, the sixth line holds as we relax the limit inferior over a subsequence to the limit inferior over the whole sequence, the penultimate line follows as $(i',j')$ is a particular pair of indices and the last line follows from the definition of $\underline{C}^{\infty}(\sC_{i}\|\sC_{j})$.

Combining Eqs.~\eqref{eq: liminf lower bound 1},~\eqref{eq: liminf lower bound 2} and~\eqref{eq: liminf lower bound 3}, we have
\begin{align}
\liminf_{n\to \infty} - \frac{1}{n}\log P_{e,\min}(\{\sC_{i,n}\}_{i=1}^r)  
   \geq  & 
   \min_{i\neq j}
  \underline{C}^{\infty}(\sC_{i}\|\sC_{j}).
\end{align}
This completes the proof of the lower bound.

\subsection{Proof of upper bound}

The proof of the upper bound follows a similar approach to~\cite{nussbaum2010asymptotically}, where a test for multiple hypotheses is converted into a modified test for a binary hypothesis. This reduction allows us to apply Theorem~\ref{thm: chernoff bound for two sets} to bound the error exponent for each binary case.

Let $\{M_{i,n}\}_{i=1}^r$ be a quantum test for $\{\sC_{i,n}\}_{i=1}^r$. For any two fixed indices $1\leq i<j\leq r$, let $A_n,B_n \in \PSD$ such that $A_n + B_n = I - M_{i,n} - M_{j,n}$. Then we consider the modified test 
\begin{align}
M_{i,n}' = M_{i,n} + A_n, \quad M_{j,n}' = M_{j,n} + B_n,
\end{align}
for $\sC_{i,n}$ and $\sC_{j,n}$. We have 
\begin{align}
    P_{e}(\{M_{i,n}\}_{i=1}^r, & \{\sC_{i,n}\}_{i=1}^r)\\
    & = \sum_{i=1}^r \pi_i \sup_{\rho_{i,n} \in \sC_{i,n}} \tr[\rho_{i,n} (I - M_{i,n})]\\
    & \geq \pi_i \sup_{\rho_{i,n} \in \sC_{i,n}} \tr[\rho_{i,n} (I - M_{i,n})] + \pi_j \sup_{\rho_{j,n} \in \sC_{j,n}} \tr[\rho_{j,n} (I - M_{j,n})]\\
    & \geq \pi_i \sup_{\rho_{i,n} \in \sC_{i,n}} \tr[\rho_{i,n} (I - M'_{i,n})] + \pi_j \sup_{\rho_{j,n} \in \sC_{j,n}} \tr[\rho_{j,n} (I - M'_{j,n})]\\
    & \geq 2\min\{\pi_i,\pi_j\} P_{e}(\{M_{i,n}', M_{j,n}'\}, \{\sC_{i,n},\sC_{j,n}\})\\
    & \geq 2\min\{\pi_i,\pi_j\} P_{e,\min}(\{\sC_{i,n},\sC_{j,n}\}),
\end{align}
where the first inequality follows by retaining only the terms corresponding to indices $i$ and $j$ in the summation; the second inequality holds as $M_{i,n}' \geq M_{i,n}$ and $M_{j,n}' \geq M_{j,n}$; the third inequality follows by considering the binary test between $\sC_{i,n}$ and $\sC_{j,n}$ with prior probabilities $\{1/2, 1/2\}$; and the last inequality holds as $\{M_{i,n}', M_{j,n}'\}$ is a valid POVM for this binary hypotheses.

As this holds for any test $\{M_{i,n}\}_{i=1}^r$, we obtain
\begin{align}
    P_{e,\min}(\{\sC_{i,n}\}_{i=1}^r) \geq 2\min\{\pi_i,\pi_j\} P_{e,\min}(\{\sC_{i,n}, \sC_{j,n}\}).
\end{align}
Taking logarithms, dividing by $n$, and considering the limit superior, we find
\begin{align}
    \limsup_{n\to \infty} - \frac{1}{n} \log P_{e,\min}(\{\sC_{i,n}\}_{i=1}^r)
     \leq \limsup_{n\to \infty} - \frac{1}{n} \log P_{e,\min}(\{\sC_{i,n}, \sC_{j,n}\})
     \leq C^\infty(\sC_i\|\sC_j),
\end{align}
where the second inequality follows from the assumptions (C2) and (C3) and Theorem~\ref{thm: chernoff bound for two sets}.
Since this holds for all $1 \leq i < j \leq r$, we conclude that
\begin{align}
    \limsup_{n\to \infty} - \frac{1}{n} \log P_{e,\min}(\{\sC_{i,n}\}_{i=1}^r)
    \leq \min_{i\neq j} C^\infty(\sC_i\|\sC_j).
\end{align}
This completes the proof of the upper bound.

\section{Maximum overlap with free states in resource theory}\label{sec: maximum overlap with free states}

The maximum overlap between a pure state $\ket{\psi}$ and a set of free states $\sF$,
\begin{align}
O_{\sF}(\psi):= \sup_{\sigma \in \sF}  \bra{\psi} \sigma \ket{\psi},
\end{align}
is a technical quantity that appears frequently in quantum resource theory~\cite{fang2020no,fang2022no,liu2019one}. Here, we provide an operational interpretation of this quantity as the optimal error exponent in symmetric hypothesis testing. 
This connection justifies the maximum overlap as a meaningful resource quantifier. Furthermore, it provides another explicit example where the quantum Chernoff divergence between sets of quantum states is not additive, thereby illustrating the necessity of regularization for the quantum Chernoff divergence.

\begin{boxtheorem}\label{thm: maximum overlap with free states}
Let $\ket{\psi}\bra{\psi}$ be a pure state and $\sF \subseteq \density$ be a convex set of quantum states. Then
\begin{align}
    C(\ket{\psi}\bra{\psi}\|\sF) =  -\log O_{\sF}(\psi).
\end{align}
\end{boxtheorem}

\begin{proof}
    We have the following chain of equalities:
    \begin{align}
        C(\ket{\psi}\bra{\psi}\|\sF) & =  \inf_{\sigma \in \sF} C(\ket{\psi}\bra{\psi}\|\sigma)\\
        & = -\log \sup_{\sigma \in \sF} \inf_{\alpha \in [0,1]}  \tr[\ket{\psi}\bra{\psi}^{\alpha} \sigma^{1-\alpha}]\\
        & = -\log \sup_{\sigma \in \sF}   \inf_{\alpha \in [0,1]} \tr[\ket{\psi}\bra{\psi}\sigma^{1-\alpha}]\\
        & = -\log  \sup_{\sigma \in \sF} \tr[\ket{\psi}\bra{\psi}\sigma]\\
        & = -\log O_{\sF}(\psi),
    \end{align}
    where the first, second and last equalities follow from definitions, the fourth equality follows from the fact that $\sigma^\alpha \geq \sigma^\beta$ if $\alpha \leq \beta$. This concludes the proof.
\end{proof}

We now present explicit examples from several resource theories in which the maximum overlap with free states can be computed exactly. These values directly determine the optimal error exponent for quantum hypothesis testing between a pure state and a set of free states.

\paragraph{Resource theory of magic.} The maximum overlap with free states is a fundamental quantity in the resource theory of magic~\cite{bravyi2019simulation}. In this context, the set of free states $\sF$ is typically chosen as the set of stabilizer states on $n$ qubits, denoted $\STAB_n$, which is stable under tensor product. The maximum overlap $O_{\STAB}(\psi)$ is closely related to the stabilizer rank and stabilizer extent—key quantities in fault-tolerant quantum computation.\footnote{In~\cite{bravyi2019simulation}, the maximum overlap is defined with respect to pure stabilizer states. However, since the objective function is linear, maximizing over the convex hull (i.e., all stabilizer states) yields the same value.} For certain states, such as the magic T-state $\ket{T}:= (\ket{0} + e^{i\pi/4}\ket{1})/\sqrt{2}$, the maximum overlap can be computed explicitly~\cite{bravyi2019simulation}:
\begin{align}
    O_{\STAB}(\ket{T}\bra{T}^{\ox m}) = (4-2\sqrt{2})^{-m}.
\end{align}
This leads to the quantum Chernoff divergence,
\begin{align}
    C(\ket{T}\bra{T}\|\STAB_1) = C^\infty(\{\ket{T}\bra{T}^{\ox n}\}_{n\in \NN}\|\{\STAB_n\}_{n\in \NN}) = \log(4-2\sqrt{2}).
\end{align}
It is also known from~\cite[Section 6.2]{bravyi2019simulation} that there exist quantum states for which the maximum overlap with stabilizer states is \emph{not multiplicative}. Consequently, the quantum Chernoff divergence between two sets is not additive in general.

\paragraph{Resource theory of coherence.} In the resource theory of coherence, the set of free states is the set of incoherent states $\cI_n = \{\rho \in \density(\cH^{\ox n}): \rho = \text{diag}(\rho)\}$, i.e., states diagonal in a fixed basis~\cite{fang2018probabilistic,regula2018one,hayashi2021finite}. Let $\ket{\psi} = \sum_{i=1}^d a_i \ket{i}$ and $\ket{\psi}\bra{\psi} \in \density(\cH)$. Then,
\begin{align}
    O_{\cI_1}(\ket{\psi}\bra{\psi}) & = \max_{\sigma \in \cI_1} \bra{\psi} \sigma \ket{\psi}  = \max_{\sigma \in \cI_1} \sum_{i=1}^d |a_i|^2 \sigma_{i}
\end{align}
where $\sigma = \sum_{i=1}^d \sigma_{i} \ket{i}\bra{i}$, $\sigma_{i} \geq 0$, and $\sum_{i=1}^d \sigma_{i} = 1$. Note that the objective function is an average of the probability vector $(|a_1|^2, |a_2|^2, \ldots, |a_d|^2)$. So it is no greater than $\max_i |a_i|^2$. This value can be achieved by choosing $\sigma_{i_{\max}}=1$ for $i_{\max}=\arg\max_i |a_i|^2$ and $\sigma_i=0$ otherwise. Therefore, we have
\begin{align}
    O_{\cI_1}(\ket{\psi}\bra{\psi}) = \max_{i} |a_i|^2,
\end{align}
which is the infinity norm of the probability vector $(|a_1|^2, |a_2|^2, \ldots, |a_d|^2)$. This leads to the quantum Chernoff divergence
\begin{align}
    C(\ket{\psi}\bra{\psi}\|\cI_1) = C^\infty(\{\ket{\psi}\bra{\psi}^{\ox n}\}_{n\in \NN}\|\{\cI_n\}_{n\in \NN}) = -\log \max_{i} |a_i|^2.
\end{align}

More generally, the Petz \Renyi divergence of a general quantum state $\rho$ with respect to the set of incoherent states $D_{\Petz,\alpha}(\rho\|\cI_1)$ is additive and has a closed-form expression as~\cite{Chitambar_2016comparison}
\begin{align}
    D_{\Petz,\alpha}(\rho\|\cI_1) = \frac{\alpha}{\alpha-1} \log \tr\left[(\text{diag}(\rho^\alpha))^{1/\alpha}\right].
\end{align}
This implies that the Chernoff divergence is given by 
\begin{align}\label{eq: chernoff coherence}
    C(\rho\|\cI_1) = C^\infty(\{\rho^{\ox n}\}_{n \in \NN}\|\{\cI_n\}_{n\in \NN}) = \sup_{\alpha \in (0,1)} -\alpha \log \tr\left[(\text{diag}(\rho^\alpha))^{1/\alpha}\right].
\end{align}

\paragraph{Resource theory of entanglement.} In the resource theory of entanglement, the standard resource is the maximally entangled state $\ket{\Phi_m} := \frac{1}{\sqrt{m}} \sum_{i=1}^m \ket{ii}$~\cite{fang2019non,regula2019one}. The standard sets of free states are the separable states $\SEP$ and the positively partial transpose states $\PPT$, with the inclusion $\SEP_n(A^n:B^n) \subseteq \PPT_n(A^n:B^n)$. The maximum overlap of $\ket{\Phi_m}$ with these sets is given by~\cite[Lemma 2]{Rains1999bound}:
\begin{align}
    O_\SEP(\Phi_m) = O_\PPT(\Phi_m) = \frac{1}{m},
\end{align}
where the maximum is achieved, for example, by the product state $\ket{0}\bra{0}_A \otimes \ket{0}\bra{0}_B$. Consequently, the quantum Chernoff divergence is
\begin{align}
    C(\Phi_m\|\SEP_1) & = C^\infty(\{\Phi_m^{\otimes n}\}_{n\in \NN}\|\{\SEP_n\}_{n\in \NN}) = \log m,\\
    C(\Phi_m\|\PPT_1) & = C^\infty(\{\Phi_m^{\otimes n}\}_{n\in \NN}\|\{\PPT_n\}_{n\in \NN}) = \log m.
\end{align}

Due to the correspondence between a maximally correlated state $\rho_{\text{mc}} := \sum_{i,j} \rho_{ij} \ket{ii}\bra{jj}$ and a coherent state $\rho = \sum_{i,j} \rho_{ij} \ket{i}\bra{j}$~\cite[Corollary 1]{Zhu_2017}, we obtain an analog to Eq.~\eqref{eq: chernoff coherence}:
\begin{align}
    C(\rho_{\text{mc}}\|\SEP_1) & = C^\infty(\{\rho_{\text{mc}}^{\otimes n}\}_{n \in \NN} \|\{\SEP_n\}_{n\in \NN}) = \sup_{\alpha \in (0,1)} -\alpha \log \tr\left[(\text{diag}(\rho_{\text{mc}}^\alpha))^{1/\alpha}\right],\\
    C(\rho_{\text{mc}}\|\PPT_1) & = C^\infty(\{\rho_{\text{mc}}^{\otimes n}\}_{n \in \NN} \|\{\PPT_n\}_{n\in \NN}) = \sup_{\alpha \in (0,1)} -\alpha \log \tr\left[(\text{diag}(\rho_{\text{mc}}^\alpha))^{1/\alpha}\right].
\end{align}

\section{Discussion}
\label{sec: discussions}

We have established generalized quantum Chernoff bounds for the discrimination of multiple sets of quantum states, thereby extending the classical and quantum Chernoff bounds to the general setting of composite and correlated quantum hypotheses. Our main results show that the optimal asymptotic error exponent for discriminating between stable sequences of convex sets of quantum states is exactly given by the regularized Chernoff divergence between the sets. The minimal assumptions required ensure that our results are broadly applicable to a wide range of quantum information tasks. Furthermore, we have provided explicit constructions of the optimal measurement for binary composite hypotheses and given an operational interpretation of the maximum overlap with free states in quantum resource theories.

Several open questions and future directions remain. While the optimal exponent in the asymmetric (Stein's) setting can be efficiently computed despite the need for regularization~\cite{fang2025efficient}, it remains open whether efficient algorithms exist for computing the regularized Chernoff divergence in the symmetric setting given similar assumptions. As noted in Remark~\ref{rem: computability of Chernoff sets}, the Chernoff divergence can be efficiently computed for fixed $n$, but the rate of convergence of the regularized Chernoff divergence as $n$ increases is not well understood.  
Additionally, our construction of the optimal test for binary composite hypotheses assumes that the difference between the closest states is full rank. It would be interesting to determine whether this assumption can be relaxed, perhaps by appealing to continuity arguments or alternative analytical techniques.

Another direction concerns the equivalence between adaptive and nonadaptive strategies in adversarial quantum channel discrimination. It was shown in~\cite{fang2025adversarial} that, in the asymmetric hypothesis testing setting, adaptive and nonadaptive strategies yield the same optimal error exponent. Whether this equivalence persists in the symmetric setting remains an open question. In particular, based on the quantum Chernoff bound established in this work, this question reduces to whether the regularized Chernoff divergence between sequences of sets of quantum states generated by adaptive and nonadaptive strategies coincide.

\paragraph{Acknowledgements.} We thank Li Gao, Ke Li, Zhiwen Lin, Yupan Liu, Andre Milzarek, Milan Mosonyi, Ryuji Takagi and Junchi Yang for helpful discussions. K.F. is supported by the National Natural Science Foundation of China (Grant No. 92470113 and 12404569), the Shenzhen Science and Technology Program (Grant No. JCYJ20240813113519025), the Shenzhen Fundamental Research Program (Grant No. JCYJ20241202124023031), and the University Development Fund (Grant No. UDF01003565). M.H. is supported in part by the National Natural Science Foundation of China (Grant No. 62171212). Both authors are supported by the General R\&D Projects of 1+1+1 CUHK-CUHK(SZ)-GDST Joint Collaboration Fund (Grant No. GRDP2025-022).

\bibliographystyle{alpha_abbrv}
\bibliography{Bib}

\newcommand{\etalchar}[1]{$^{#1}$}
\begin{thebibliography}{ACMT{\etalchar{+}}07}

\bibitem[ACMT{\etalchar{+}}07]{audenaert2007discriminating}
K.~M. Audenaert, J.~Calsamiglia, R.~Munoz-Tapia, E.~Bagan, L.~Masanes, A.~Acin, and F.~Verstraete.
\newblock {Discriminating states: The quantum Chernoff bound}.
\newblock {\em Physical Review Letters}, 98(16):160501, 2007.

\bibitem[ANSV08]{audenaert2008asymptotic}
K.~M. Audenaert, M.~Nussbaum, A.~Szko{\l}a, and F.~Verstraete.
\newblock Asymptotic error rates in quantum hypothesis testing.
\newblock {\em Communications in Mathematical Physics}, 279(1):251--283, 2008.

\bibitem[BBC{\etalchar{+}}19]{bravyi2019simulation}
S.~Bravyi, D.~Browne, P.~Calpin, E.~Campbell, D.~Gosset, and M.~Howard.
\newblock Simulation of quantum circuits by low-rank stabilizer decompositions.
\newblock {\em Quantum}, 3:181, 2019.

\bibitem[BBH21]{berta2021composite}
M.~Berta, F.~G. Brandão, and C.~Hirche.
\newblock On composite quantum hypothesis testing.
\newblock {\em Communications in Mathematical Physics}, 385(1):55--77, 2021.

\bibitem[BDB23]{ben2023new}
S.~Ben-David and E.~Blais.
\newblock A new minimax theorem for randomized algorithms.
\newblock {\em Journal of the ACM}, 70(6):1--58, 2023.

\bibitem[Ber09]{bertsekas2009convex}
D.~Bertsekas.
\newblock {\em Convex optimization theory}, volume~1.
\newblock Athena Scientific, 2009.

\bibitem[BHLP20]{brandao2020adversarial}
F.~G. Brandão, A.~W. Harrow, J.~R. Lee, and Y.~Peres.
\newblock Adversarial hypothesis testing and a quantum {Stein’s} lemma for restricted measurements.
\newblock {\em IEEE Transactions on Information Theory}, 66(8):5037--5054, 2020.

\bibitem[BP10]{brandao2010generalization}
F.~G. Brandao and M.~B. Plenio.
\newblock A generalization of quantum {Stein}'s lemma.
\newblock {\em Communications in Mathematical Physics}, 295(3):791--828, 2010.

\bibitem[CG16]{Chitambar_2016comparison}
E.~Chitambar and G.~Gour.
\newblock Comparison of incoherent operations and measures of coherence.
\newblock {\em Physical Review A}, 94:052336, Nov 2016.

\bibitem[Che52]{chernoff1952measure}
H.~Chernoff.
\newblock A measure of asymptotic efficiency for tests of a hypothesis based on the sum of observations.
\newblock {\em The Annals of Mathematical Statistics}, pages 493--507, 1952.

\bibitem[Che00]{chen2000GeneralizationGartnerEllisTheorem}
P.-N. Chen.
\newblock Generalization of {{Gartner-Ellis}} theorem.
\newblock {\em IEEE Transactions on Information Theory}, 46(7):2752--2760, November 2000.

\bibitem[DWH25]{DWH25}
A.~Dasgupta, N.~A. Warsi, and M.~Hayashi.
\newblock Universal tester for multiple independence testing and classical-quantum arbitrarily varying multiple access channel.
\newblock {\em IEEE Transactions on Information Theory}, 71(5):3719--3765, 2025.

\bibitem[DZ10]{Dembo2010}
A.~Dembo and O.~Zeitouni.
\newblock {\em Large Deviations Techniques and Applications}.
\newblock Springer Berlin Heidelberg, 2010.

\bibitem[ET99]{ekeland1999convex}
I.~Ekeland and R.~Temam.
\newblock {\em Convex analysis and variational problems}.
\newblock SIAM, 1999.

\bibitem[FFF24]{fang2024generalized}
K.~Fang, H.~Fawzi, and O.~Fawzi.
\newblock Generalized quantum asymptotic equipartition.
\newblock {\em arXiv preprint arXiv:2411.04035}, 2024.

\bibitem[FFF25a]{fang2025adversarial}
K.~Fang, H.~Fawzi, and O.~Fawzi.
\newblock Adversarial quantum channel discrimination.
\newblock {\em arXiv preprint arXiv:2506.03060}, 2025.

\bibitem[FFF25b]{fang2025efficient}
K.~Fang, H.~Fawzi, and O.~Fawzi.
\newblock Efficient approximation of regularized relative entropies and applications.
\newblock {\em arXiv preprint arXiv:2502.15659}, 2025.

\bibitem[FL20]{fang2020no}
K.~Fang and Z.-W. Liu.
\newblock No-go theorems for quantum resource purification.
\newblock {\em Physical Review Letters}, 125(6):060405, 2020.

\bibitem[FL22]{fang2022no}
K.~Fang and Z.-W. Liu.
\newblock No-go theorems for quantum resource purification: New approach and channel theory.
\newblock {\em PRX Quantum}, 3(1):010337, 2022.

\bibitem[FR06]{farkas2006potential}
B.~Farkas and S.~G. R{\'e}v{\'e}sz.
\newblock Potential theoretic approach to rendezvous numbers.
\newblock {\em Monatshefte f{\"u}r mathematik}, 148(4):309--331, 2006.

\bibitem[FWL{\etalchar{+}}18]{fang2018probabilistic}
K.~Fang, X.~Wang, L.~Lami, B.~Regula, and G.~Adesso.
\newblock Probabilistic distillation of quantum coherence.
\newblock {\em Physical Review Letters}, 121(7):070404, 2018.

\bibitem[FWTD19]{fang2019non}
K.~Fang, X.~Wang, M.~Tomamichel, and R.~Duan.
\newblock Non-asymptotic entanglement distillation.
\newblock {\em IEEE Transactions on Information Theory}, 65(10):6454--6465, 2019.

\bibitem[Gur03]{gurvits2003classical}
L.~Gurvits.
\newblock Classical deterministic complexity of edmonds' problem and quantum entanglement.
\newblock In {\em Proceedings of the thirty-fifth annual ACM symposium on Theory of computing}, pages 10--19, 2003.

\bibitem[Hay14]{hayashi2014large}
M.~Hayashi.
\newblock {Large deviation analysis for quantum security via smoothing of R{\'e}nyi entropy of order 2}.
\newblock {\em IEEE Transactions on Information Theory}, 60(10):6702--6732, 2014.

\bibitem[Hay17]{hayashi2017quantum}
M.~Hayashi.
\newblock Quantum information theory.
\newblock {\em Graduate Texts in Physics, Springer}, 2017.

\bibitem[Hel69]{helstrom1969quantum}
C.~W. Helstrom.
\newblock Quantum detection and estimation theory.
\newblock {\em Journal of Statistical Physics}, 1(2):231--252, 1969.

\bibitem[HFW21]{hayashi2021finite}
M.~Hayashi, K.~Fang, and K.~Wang.
\newblock Finite block length analysis on quantum coherence distillation and incoherent randomness extraction.
\newblock {\em IEEE Transactions on Information Theory}, 67(6):3926--3944, 2021.

\bibitem[HI25]{hayashi2025entanglement}
M.~Hayashi and Y.~Ito.
\newblock Entanglement measures for detectability.
\newblock {\em IEEE Transactions on Information Theory}, 71(6):4385–4405, April 2025.

\bibitem[HMO07]{hiai2007large}
F.~Hiai, M.~Mosonyi, and T.~Ogawa.
\newblock Large deviations and chernoff bound for certain correlated states on a spin chain.
\newblock {\em Journal of Mathematical Physics}, 48(12), 2007.

\bibitem[HMO08]{hiai2008error}
F.~Hiai, M.~Mosonyi, and T.~Ogawa.
\newblock Error exponents in hypothesis testing for correlated states on a spin chain.
\newblock {\em Journal of Mathematical Physics}, 49(3), 2008.

\bibitem[Hol72]{holevo1972analog}
A.~Holevo.
\newblock An analog of the theory of statistical decisions in noncommutative probability theory.
\newblock {\em Trans. Moscow Math. Soc}, 26(1972):133--149, 1972.

\bibitem[HSF24]{he2024qics}
K.~He, J.~Saunderson, and H.~Fawzi.
\newblock {QICS}: {Q}uantum information conic solver, 2024.

\bibitem[HT16]{hayashi2016correlation}
M.~Hayashi and M.~Tomamichel.
\newblock {Correlation detection and an operational interpretation of the {R{\'e}nyi} mutual information}.
\newblock {\em Journal of Mathematical Physics}, 57(10), 2016.

\bibitem[HY25]{hayashi2024generalized}
M.~Hayashi and H.~Yamasaki.
\newblock {The generalized quantum Stein’s lemma and the second law of quantum resource theories}.
\newblock {\em Nature Physics}, 2025.
\newblock \url{https://doi.org/10.1038/s41567-025-03047-9}.

\bibitem[LBT19]{liu2019one}
Z.-W. Liu, K.~Bu, and R.~Takagi.
\newblock One-shot operational quantum resource theory.
\newblock {\em Physical Review Letters}, 123(2):020401, 2019.

\bibitem[Li16]{li2016discriminating}
K.~Li.
\newblock {Discriminating quantum states: The multiple Chernoff distance}.
\newblock {\em The Annals of Statistics}, 44(4), August 2016.

\bibitem[LLS{\etalchar{+}}23]{Leditzky_2023}
F.~Leditzky, D.~Leung, V.~Siddhu, G.~Smith, and J.~A. Smolin.
\newblock The platypus of the quantum channel zoo.
\newblock {\em IEEE Transactions on Information Theory}, 69(6):3825–3849, June 2023.

\bibitem[LR22]{lehmann-2022}
E.~Lehmann and J.~P. Romano.
\newblock {\em {Testing statistical hypotheses}}.
\newblock 1 2022.

\bibitem[MH11]{mosonyi2011quantum}
M.~Mosonyi and F.~Hiai.
\newblock {On the quantum R{\'e}nyi relative entropies and related capacity formulas}.
\newblock {\em IEEE Transactions on Information Theory}, 57(4):2474--2487, 2011.

\bibitem[MH23]{mosonyi2023some}
M.~Mosonyi and F.~Hiai.
\newblock {Some continuity properties of quantum R{\'e}nyi divergences}.
\newblock {\em IEEE Transactions on Information Theory}, 70(4):2674--2700, 2023.

\bibitem[MO15]{mosonyi2015two}
M.~Mosonyi and T.~Ogawa.
\newblock Two approaches to obtain the strong converse exponent of quantum hypothesis testing for general sequences of quantum states.
\newblock {\em IEEE Transactions on Information Theory}, 61(12):6975--6994, 2015.

\bibitem[MSW22]{Mosonyi_2022}
M.~Mosonyi, Z.~Szilagyi, and M.~Weiner.
\newblock On the error exponents of binary state discrimination with composite hypotheses.
\newblock {\em IEEE Transactions on Information Theory}, 68(2):1032–1067, February 2022.

\bibitem[NS09]{Nussbaum2009}
M.~Nussbaum and A.~Szko{\l}a.
\newblock {The Chernoff lower bound for symmetric quantum hypothesis testing}.
\newblock {\em The Annals of Statistics}, 37(2):1040--1057, 2009.

\bibitem[NS10]{nussbaum2010asymptotically}
M.~Nussbaum and A.~Szko{\l}a.
\newblock Asymptotically optimal discrimination between pure quantum states.
\newblock In {\em Conference on Quantum Computation, Communication, and Cryptography}, pages 1--8. Springer, 2010.

\bibitem[Rai99]{Rains1999bound}
E.~M. Rains.
\newblock Bound on distillable entanglement.
\newblock {\em Physical Review A}, 60:179--184, Jul 1999.

\bibitem[RFWA18]{regula2018one}
B.~Regula, K.~Fang, X.~Wang, and G.~Adesso.
\newblock One-shot coherence distillation.
\newblock {\em Physical Review Letters}, 121(1):010401, 2018.

\bibitem[RFWG19]{regula2019one}
B.~Regula, K.~Fang, X.~Wang, and M.~Gu.
\newblock One-shot entanglement distillation beyond local operations and classical communication.
\newblock {\em New Journal of Physics}, 21(10):103017, 2019.

\bibitem[Tom16]{Tomamichel2015b}
M.~Tomamichel.
\newblock {\em {Quantum Information Processing with Finite Resources}}, volume~5 of {\em SpringerBriefs in Mathematical Physics}.
\newblock Springer International Publishing, Cham, 2016.

\bibitem[WT24]{watanabe2024black}
K.~Watanabe and R.~Takagi.
\newblock Black box work extraction and composite hypothesis testing.
\newblock {\em Physical Review Letters}, 133(25):250401, 2024.

\bibitem[ZHC17]{Zhu_2017}
H.~Zhu, M.~Hayashi, and L.~Chen.
\newblock Coherence and entanglement measures based on rényi relative entropies.
\newblock {\em Journal of Physics A: Mathematical and Theoretical}, 50(47):475303, November 2017.

\end{thebibliography}

\end{document}